\documentclass{article}
\usepackage[sort&compress,square,sort,comma,numbers]{natbib}
\usepackage[margin=1.25in]{geometry}

\usepackage{amsmath,amssymb,amsthm}
\usepackage{hyperref}
\hypersetup{hidelinks}
\usepackage[binary-units]{siunitx}

\newtheorem{theorem}{Theorem}[section]

\newtheorem{lemma}[theorem]{Lemma}
\theoremstyle{definition}
\newtheorem{proposition}[theorem]{Proposition}
\newtheorem{assumption}[theorem]{Assumption}
\newtheorem{definition}[theorem]{Definition}
\theoremstyle{example}
\newtheorem{example}[theorem]{Example}

\newcommand{\Name}[1]{\textsc{#1}}
\newcommand{\Pbft}{\Name{Pbft}}
\newcommand{\Paxos}{\Name{Paxos}}
\newcommand{\HS}{\Name{HotStuff}}

\newcommand{\CB}{\Name{Cerberus}}
\newcommand{\cCB}{\Name{CCerberus}}
\newcommand{\oCB}{\Name{OCerberus}}
\newcommand{\pCB}{\Name{PCerberus}}

\newcommand{\Object}{o}
\newcommand{\ID}[1]{\fnname{id}(#1)}
\newcommand{\Owner}[1]{\fnname{owner}(#1)}
\newcommand{\Client}{c}
\newcommand{\Transaction}{\tau}
\newcommand{\Request}[2]{\langle#1\rangle_{#2}}
\newcommand{\Inputs}[1]{\fnname{Inputs}(#1)}
\newcommand{\Outputs}[1]{\fnname{Outputs}(#1)}
\newcommand{\Objects}[1]{\fnname{Objects}(#1)}

\newcommand{\NonFaulty}[1]{\mathcal{G}(#1)}
\newcommand{\Faulty}[1]{\mathcal{F}(#1)}
\newcommand{\n}[1]{\mathbf{n}_{#1}}
\newcommand{\nf}[1]{\mathbf{g}_{#1}}
\newcommand{\f}[1]{\mathbf{f}_{#1}}
\newcommand{\z}{\mathbf{z}}
\newcommand{\Replica}[1][r]{\textnormal{\textsc{#1}}}
\newcommand{\Replicas}{\mathfrak{R}}
\newcommand{\Shard}{\mathcal{S}}
\newcommand{\Shards}[1]{\fnname{shards}(#1)}
\newcommand{\ObjectShard}[1]{\fnname{shard}(#1)}
\newcommand{\Primary}[1]{\mathcal{P}(#1)}
\newcommand{\TOrder}{\mathop{{\prec}}}

\newcommand{\rn}{\rho}
\newcommand{\View}[1][v]{#1}

\newcommand{\MName}[1]{\textnormal{\texttt{#1}}}
\newcommand{\Message}[2]{\langle\MName{#1} : #2\rangle}

\newcommand{\abs}[1]{\lvert #1 \rvert}
\newcommand{\union}{\cup}
\newcommand{\intersect}{\cap}
\newcommand{\difference}{\setminus}
\newcommand{\fnname}[1]{\mathop{\texttt{#1}}}

\newcommand{\lfref}[2]{Line~\ref{#1:#2} of Figure~\ref{#1}}

\newenvironment{requirement}{
    \begin{enumerate}
    
}{
    \end{enumerate}
}

\usepackage[noend]{algorithmic}

\newenvironment{myprotocol}{
    \hrule
    \smallskip
    \begin{algorithmic}[1]
        \newcommand{\SPACE}{\item[]}
        \renewcommand{\algorithmiccomment}[1]{\{ \emph{##1} \}}

        \makeatletter
            \newcommand{\EVENT}[2][default]{\STATE \textbf{event} ##2 \textbf{do}\ifthenelse{\equal{##1}{default}}{}{\ \algorithmiccomment{##1}}\begin{ALC@g}}
            \newcommand{\ENDEVENT}{\end{ALC@g}}
        \makeatother
}{
    \end{algorithmic}
    \smallskip
    \hrule
}

\usepackage{tikz,pgfplots,pgfplotstable}
\usetikzlibrary{arrows.meta,decorations.pathreplacing}
\tikzset{
    plot/.append style={baseline,scale=0.535},
    dot/.append style={circle,scale=0.35,draw=black,fill=black},
    label/.append style={align=center,font=\strut\footnotesize},
    >=Stealth
}
\definecolor{colA}{RGB}{230,159,0}
\definecolor{colB}{RGB}{86,180,233}
\definecolor{colC}{RGB}{0,158,115}
\definecolor{colD}{RGB}{240,228,66}
\definecolor{colE}{RGB}{0,114,178}
\definecolor{colF}{RGB}{213,94,0}
\definecolor{colG}{RGB}{204,121,167}

\pgfplotscreateplotcyclelist{mycyclelist}{         
    solid,colD,mark=*\\          
    solid,colE,mark=*\\          
    solid,colG,mark=*\\
    solid,colA,mark=*\\          
    solid,colC,mark=*\\          
    solid,colB,mark=*\\          
    solid,colC,mark=*\\          
    solid,colF,mark=*\\          
}

\pgfplotsset{
    compat=1.14,
    tick label style={font=\large},
    legend style={font=\Large,cells={anchor=west}},
    title style={font=\Large},
    label style={font=\Large},
    width=300pt,
    height=155pt,
    every axis/.append style={
        ylabel near ticks,
        xlabel near ticks,
        mark size=1pt,
        cycle list name=mycyclelist,
        font=\Large,
        enlargelimits=0.1
    }
}

\newcommand{\plotPerformance}[3]{
    \begin{tikzpicture}[plot]
        \begin{axis}[title={Performance (\SI{#2}{\text{obj}\per\text{txn}})},
                     ylabel={Throughput (\si{\text{txn}\per\second})},xlabel={Shards},
                     xmode=log,log base x={2},ymode=log,log base y={10},xmin=1,xmax=16384,xtick={1,4,16,64,256,1024,4096,16384},#3]
            \addplot table[x={num_shards}, y={ccb_tput}] {#1};
            \addplot table[x={num_shards}, y={ocb_tput}] {#1};
            \addplot table[x={num_shards}, y={pcb_tput}] {#1};

        \end{axis}
    \end{tikzpicture}%
}
\newcommand{\plotSizes}[3]{
    \begin{tikzpicture}[plot]
        \begin{axis}[title={#2},ylabel={Steps},xlabel={Shards},
                     xmode=log,log base x={2},xmin=1,xmax=16384,xtick={1,4,16,64,256,1024,4096,16384},#3,
        y tick label style={
            /pgf/number format/precision=1,
            /pgf/number format/fixed,
            /pgf/number format/fixed zerofill
        }]
            \addplot table[x={num_shards}, y={#1}] {\dataTwo};
            \addplot table[x={num_shards}, y={#1}] {\dataFour};
            \addplot table[x={num_shards}, y={#1}] {\dataEight};
            \addplot table[x={num_shards}, y={#1}] {\dataSixteen};
            \addplot table[x={num_shards}, y={#1}] {\dataThirtyTwo};
            \addplot table[x={num_shards}, y={#1}] {\dataSixtyFour};
        \end{axis}
    \end{tikzpicture}%
}

\title{\CB{}: Minimalistic Multi-shard Byzantine-resilient Transaction Processing}
\author{
\begin{tabular}{c@{\qquad}c}
Jelle Hellings\footnotemark[1] & Daniel P.\ Hughes\footnotemark[2]\\
Joshua Primero\footnotemark[2]& Mohammad Sadoghi\footnotemark[1]
\end{tabular}}
\date{\normalsize\footnotemark[1] Exploratory Systems Lab, Department of Computer Science\\University of California, Davis, CA, 95616-8562, USA\\
    \footnotemark[2] Radix DLT Ltd, Argyle Works, 29-31 Euston Road, London, NW1 2SD}

\begin{document}

\maketitle

\begin{abstract}
To enable high-performance and scalable blockchains, we need to step away from traditional consensus-based fully-replicated designs. One direction is to explore the usage of \emph{sharding} in which we partition the managed dataset over many shards that---independently---operate as blockchains. Sharding requires an efficient  fault-tolerant primitive for the ordering and execution of multi-shard transactions, however.

In this work, we seek to design such a primitive suitable for distributed ledger networks with high transaction throughput. To do so, we propose \CB{}, a set of minimalistic primitives for processing single-shard and multi-shard UTXO-like transactions. \CB{} aims at maximizing parallel processing at shards while minimizing coordination within and between shards. First, we propose \emph{Core-\CB{}}, that uses strict environmental requirements to enable simple yet powerful multi-shard transaction processing. In our intended UTXO-environment, Core-\CB{} will operate \emph{perfectly} with respect to all transactions proposed and approved by well-behaved clients, but does not provide any guarantees for other transactions.

To also support more general-purpose environments, we propose \emph{two} generalizations of Core-\CB{}: we propose \emph{Optimistic-\CB{}}, a protocol that does not require any additional coordination phases in the well-behaved optimistic case, while requiring intricate coordination when recovering from attacks; and we propose \emph{Pessimistic-\CB{}}, a protocol that adds sufficient coordination to the well-behaved case of Core-\CB{}, allowing it to operate in a general-purpose fault-tolerant environments without significant costs to recover from attacks. Finally, we compare the three protocols, showing their potential scalability and high transaction throughput in practical environments.
\end{abstract}

\section{Introduction}

The advent of blockchain applications and technology has rejuvenated interest of companies, governments, and developers in resilient distributed fully-replicated systems and the distributed ledger technology (DLT) that powers them. Indeed, in the last decade we have seen a surge of interest in reimagining systems and build them using DLT networks. Examples can be found in the financial and banking sector~\cite{impactblock,hypereal,promiseblock}, IoT~\cite{blockchain_iot}, health care~\cite{blockhealthfac,blockhealthover}, supply chain tracking, advertising, and in databases~\cite{caper,hyperledger,blockchaindb,blockmeetdb,blockplane}. This wide interest is easily explained, as blockchains promise to improve resilience, while enabling the federated management of data by many participants.

To illustrate this, we look at the financial sector. Current traditional banking infrastructure is often rigid, slow, and creates substantial frictional costs. It is estimated that the yearly cost of transactional friction alone is \$71 billion~\cite{ehes} in the financial sector, creating a strong desire for alternatives. This sector is a perfect match for DLT, as it enables systems that manage digital assets and financial transactions in more flexible, fast, and open federated infrastructures that eliminate the friction caused by individual private databases maintained by banks and financial services providers. Consequently, it is expected that a large part of the financial sector will move towards DLT~\cite{cftc}.
 
At the core of DLT is the \emph{replicated state} maintained by the network in the form of a ledger of transactions. In traditional blockchains, this ledger is fully replicated among all participants using consensus protocols~\cite{blockchain_iot,wild,bit_pedigree,blockchain_dist}. For many practical use-cases, one can choose to use either permissionless consensus solutions  that are operated via economic self-incentivization through cryptocurrencies (e.g., Nakamoto consensus~\cite{bitcoin,ethereum}), or permissioned consensus solutions that require vetted participation (e.g, \Pbft{}~\cite{pbftj}). Unfortunately, the design of consensus protocols utilized by todays DLT networks are severely limited in their ability to provide the \emph{high transaction throughput} that is needed to address practical needs, e.g., in the financial and banking sector.

On the one hand, we see that permissionless solutions can easily scale to thousands of participants, but are severely limited in their transaction processing throughput. E.g., in Ethereum, a popular public permissionless DLT platform, the rapid growth of decentralized finance applications~\cite{defi} has caused its network fees to rise precipitously as participants bid for limited network capacity~\cite{defieth}, while Bitcoin can only process a few transactions per second~\cite{hypereal}. On the other hand, permissioned solutions can reach much higher throughputs, but still lack scalability as their performance is bound by the speed of individual participants.

In this paper, we focus on a fundamental solution to significantly increase the throughput of DLT that may apply to either permissionless or permissioned networks. While this paper primarily discuss this solution through the lens of permissioned networks, similar techniques apply to permissionless DLT with the necessary extensions for these kinds of networks, such as self-incentivization, Sybil attack protection, and tolerance of validator set churn. These kinds of permissionless networks are the focus of Radix, and their impetus for their original creation of the \CB{} concept that this paper will discuss.

A direction one can take to improve on the limited throughput of a DLT network, is to incorporate \emph{sharding} in their design: instead of operating a single fully-replicated consensus-based DLT network, one can partition the data in the DLT network among several \emph{shards} that each have the potential to operate mostly-independent on their data, while only requiring cooperation between shards to process transactions that affect data on several shards. In such a sharded design, transactions that only affect objects within a single shard can be processed via normal consensus (e.g., \Pbft{}). Transactions that affect objects within several shards require additional coordination, however. The choice of protocol for such multi-shard transaction processing determines greatly the scalability benefits of sharding and the overhead costs incurred by sharding. We have sketched a basic sharded design in Figure~\ref{fig:example_intro}.

\begin{figure}[t!]
\centering
    \scalebox{0.75}{
    \begin{tikzpicture}[xscale=1.25]
        \filldraw[very thick,fill=green!10,draw=green!10!black!30] (-0.35, 3.65) rectangle (8.35, 6.35);
        \filldraw[very thick,fill=orange!10,draw=orange!10!black!30] (-0.25, 3.75) rectangle (2.25, 6.25);
        \node (raf1) at (0, 6) {$\Replica[A]_1$};
        \node (raf2) at (2, 6) {$\Replica[A]_2$};
        \node (raf3) at (0, 4) {$\Replica[A]_3$};
        \node (raf4) at (2, 4) {$\Replica[A]_4$};
        \path[<->,thin] (raf1) edge (raf2) edge (raf3) edge (raf4)
                        (raf2) edge (raf3) edge (raf4)
                        (raf3) edge node[above] {\Pbft{}} (raf4);
        \node[below,align=center] at (1, 3.75) {(Objects $\Object_{1},\dots,\Object_{10}$)\strut};

        \filldraw[very thick,fill=blue!10,draw=blue!10!black!30] (5.75, 3.75) rectangle (8.25, 6.25);
        \node (ram1) at (6, 6) {$\Replica[B]_1$};
        \node (ram2) at (8, 6) {$\Replica[B]_2$};
        \node (ram3) at (6, 4) {$\Replica[B]_3$};
        \node (ram4) at (8, 4) {$\Replica[B]_4$};
        \path[<->,thin] (ram1) edge (ram2) edge (ram3) edge (ram4)
                        (ram2) edge (ram3) edge (ram4)
                        (ram3) edge node[above] {\Pbft{}} (ram4);
        \node[below,align=center] at (7, 3.75) {(Objects $\Object_{11},\dots,\Object_{20}$)\strut};

        \node[above,align=center] (ec) at (1, 6.65) {Request on $\Object_3, \Object_5$\strut\\(via \Pbft{})} edge[thick,->] (1, 6.25);
        \node[above,align=center] (ec) at (7, 6.65) {Request on $\Object_{12}, \Object_{17}$\strut\\(via \Pbft{})} edge[thick,->] (7, 6.25);
        \path[<->] (2.3, 5) edge node[above] {\CB{}} (5.7, 5);
        
        \node[above,align=center] (ec) at (4, 6.65) {Request on $\Object_2, \Object_{14}$\strut\\(via \CB{})} edge[thick,->] (4, 6.35);
    \end{tikzpicture}
    }\caption{A \emph{sharded} design in which two resilient blockchains each hold only a part of the data. Local decisions within a cluster are made via \emph{traditional \Pbft{} consensus}, whereas multi-shard transactions are processed via \CB{} (proposed in this work).}\label{fig:example_intro}
\end{figure}
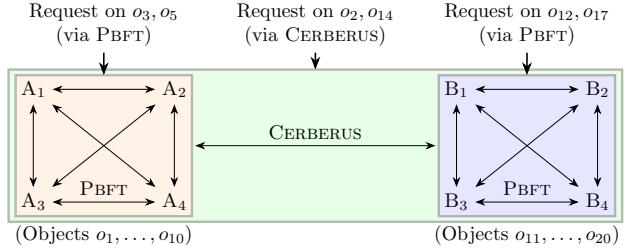

To provide multi-shard transaction processing with high throughput in practical environments with a large number of shards, including permissionless networks, Radix proposed \CB{}---a technique for performing multi-shard transactions. In this paper, we propose and analyze a family of multi-shard transaction processing protocol variants using the original \CB{} concept. To be able to adapt to the needs of specific use-cases, we propose three variants of \CB{}: Core-\CB{}, Optimistic-\CB{}, and Pessimistic-\CB{}.\footnote{The ideas underlying \CB{} was outlined in an earlier whitepaper of our Radix team available at \url{https://www.radixdlt.com/wp-content/uploads/2020/03/Cerberus-Whitepaper-v1.0.pdf}.}

First, we propose Core-\CB{} (\cCB{}), a design specialized for processing \emph{UTXO-like transactions}. \cCB{} is a simplified variant of \CB{}, and uses the strict environmental assumptions on UTXO-transactions to its advantage to yield a \emph{minimalistic} design that does as little work as possible per involved shard. Even with this minimalistic design, \cCB{} will operate \emph{perfectly} with respect to all transactions proposed and approved by well-behaved clients (although it may fail to process transactions originating from malicious clients).

Next, to also support more general-purpose environments, we propose \emph{two} generalizations of \cCB{}, namely  Optimistic-\CB{} and Pessimistic-\CB{} , that each deal with the strict environmental assumptions of \cCB{}, while preserving the minimalistic design of \cCB{}. In the design of Optimistic-\CB{} (\oCB{}), we assume that malicious behavior is rare and we optimize the normal-case operations. We do so by keeping the normal-case operations as minimalistic as possible. In specific, compared to \cCB{}, \oCB{} does not require any additional coordination phases in the well-behaved optimistic case, while still being able to lift the environmental assumptions of \cCB{}. In doing so, \oCB{} does require intricate coordination when recovering from attacks.  In the design of Pessimistic-\CB{}, we assume that malicious behavior is common and we add sufficient coordination to the normal-case operations of \cCB{} to enable a simpler and localized recovery path, allowing \pCB{} to recover from attacks at lower cost, at the expense of increased complexity in normal-case operation. Both variants we believe may be productive directions for consideration for different network deployment situations depending on desired trade-offs.

To show the strengths of each of the \CB{} protocols, we show that \CB{} can provide serializable transaction execution for UTXO-like transactions. Furthermore, we show that each of the protocol variants have excellent scalability in practice, even when exclusively dealing with multi-shard workloads. 

\paragraph*{Organization}

First, in Section~\ref{sec:prelim}, we present the terminology and notation used throughout this paper. Then, in Section~\ref{sec:correct}, we specify the \emph{correctness criteria} by which we evaluate our \CB{} multi-shard transaction processing protocols. Then, in Sections~\ref{sec:ccb},~\ref{sec:ocb}, and~\ref{sec:pcb}, we present the three variants of \CB{}, namely Core-\CB{} (\cCB{}), Optimistic-\CB{} (\oCB{}), and Pessimistic-\CB{} (\pCB{}). In Section~\ref{sec:compare}, we further analyze the practical strengths, properties, and performance of \CB{}. Then, in Section~\ref{sec:related}, we discuss related work, while we conclude on our findings in Section~\ref{sec:concl}.

\section{Preliminaries}\label{sec:prelim}

Before we proceed with our detailed presentation of \CB{}, we first introduce the system model, the sharding model, the data model, the transaction model,  and the relevant terminology and notation used throughout this paper.

\paragraph{Sharded fault-tolerant systems}
If $S$ is a set of replicas, then we write $\NonFaulty{S}$ to denote the non-faulty \emph{good replicas} in $S$ that always operate as intended, and we write $\Faulty{S} = S \difference \NonFaulty{S}$ to denote the remaining replicas in $S$ that are \emph{faulty} and can act \emph{Byzantine}, deviate from the intended operations, or even operate in coordinated malicious manners. We write $\n{S} = \abs{S}$, $\nf{S} = \abs{\NonFaulty{S}}$, and $\f{S} = \abs{S \difference \NonFaulty{S}} = \n{S} - \nf{S}$ to denote the number of replicas in $S$, good replicas in $S$, and faulty replicas in $S$, respectively.

Let $\Replicas$ be a set of replicas. In a \emph{sharded fault-tolerant system} over $\Replicas$, the replicas are partitioned into sets $\Shards{\Replicas} = \{ \Shard_0, \dots, \Shard_\z \}$ such that the replicas in $\Shard_i$, $0 \leq i \leq \z$, operate as an independent Byzantine fault-tolerant system. As each $\Shard_i$ operates as an independent fault-tolerant system, we require $\n{\Shard_i} > 3\f{\Shard_i}$, a minimal requirement to enable Byzantine fault-tolerance in an asynchronous environment~\cite{netbound,byzgen}. We assume that every shard $\Shard \in \Shards{\Replicas}$ has a unique identifier $\ID{\Shard}$.

We assume \emph{asynchronous communication}: messages can get lost, arrive with arbitrary delays, and in arbitrary order. Consequently, it is impossible to distinguish between, on the one hand, a replica that is malicious and does not send out messages, and, on the other hand, a replica that does send out proposals that get lost in the network. As such, \CB{} can only provide \emph{progress} in periods of \emph{reliable bounded-delay communication} during which all messages sent by good replicas will arrive at their destination within some maximum delay~\cite{flp,capproof}. Further, we assume that communication is \emph{authenticated}: on receipt of a message $m$ from replica $\Replica \in \Replicas$, one can determine that $\Replica$ did sent $m$ if $\Replica \in \NonFaulty{\Replicas}$. Hence, faulty replicas are able to impersonate each other, but are not able to impersonate good replicas. To provide authenticated communication under practical assumptions, we can rely on cryptographic primitives such as message authentication codes, digital signatures, or threshold signatures~\cite{rsasign,cryptobook}.

\begin{assumption}\label{ass:coordinate}
Let $\Shards{\Replicas}$ be a sharded fault-tolerant system. We assume \emph{coordinating adversaries} that can---at will---choose and control any replica $\Replica \in \Shard$ in any shard $\Shard \in \Shards{\Replicas}$ as long as, for each shard $\Shard'$, the adversaries only control up to $\f{\Shard'}$ replicas in $\Shard'$.
\end{assumption}

\paragraph{Object-dataset model}
We use the \emph{object-dataset model} in which data is modeled as a collection of \emph{objects}. Each object $\Object$ has a unique \emph{identifier} $\ID{\Object}$ and a unique \emph{owner} $\Owner{\Object}$. In the following, we assume that all owners are \emph{clients} of the system that manages these objects. The only operations that one can perform on an object are \emph{construction} and \emph{destruction}. An object cannot be recreated, as the attempted recreation of an object $\Object$ will result in a new object $\Object'$ with a distinct identifier ($\ID{\Object} \neq \ID{\Object'}$).

\paragraph{Object-dataset transactions}
Changes to object-dataset data are made via transactions requested by clients. We write $\Request{\Transaction}{\Client}$ to denote a transaction $\Transaction$ requested by a client $\Client$. We assume that all transactions are \emph{UTXO-like transactions}: a transaction $\Transaction$ first produces resources by destructing a set of \emph{input objects} and then consumes these resources in the construction of a set of \emph{output objects}. We do not rely on the exact rules regarding the production and consumption of resources, as they are highly application-specific. Given a transaction $\Transaction$, we write $\Inputs{\Transaction}$ and $\Outputs{\Transaction}$ to denote the input objects and output objects of $\Transaction$, respectively, and we write $\Objects{\Transaction} = \Inputs{\Transaction} \union \Outputs{\Transaction}$. 

\begin{assumption}\label{ass:inout}
Given a transaction $\Transaction$, we assume that one can determine $\Inputs{\Transaction}$ and $\Outputs{\Transaction}$ a-priori. Furthermore, we assume that every transaction has inputs. Hence, $\abs{\Inputs{\Transaction}} \geq 1$.
\end{assumption}

Owners of objects $\Object$ can \emph{express their support} for transactions $\Transaction$ that have $\Object$ as their input. To provide this functionality, we can rely on cryptographic primitives such as digital signatures~\cite{cryptobook}.

\begin{assumption}\label{ass:wellclient}
If an owner is well-behaved, then an expression of support cannot be forged or provided by any other party. Furthermore, a well-behaved owner of $\Object$ will only express its support for \emph{a single} transaction $\Transaction$ with $\Object \in \Inputs{\Transaction}$, as only one transaction can consume the object $\Object$, and the owner will only do so after the construction of $\Object$.
\end{assumption} 
\paragraph{Multi-shard transactions}
Let $\Object$ be an object. We assume that there is a well-defined function $\ObjectShard{\Object}$ that maps object $\Object$ to the single shard $\Shard \in \Shards{\Replicas}$ that is responsible for maintaining $\Object$. Given a transaction $\Transaction$, we write \[\Shards{\Transaction} = \{ \ObjectShard{\Object} \mid \Object \in \Objects{\Transaction} \}\] to denote the shards that are affected by $\Transaction$. We say that $\Transaction$ is a \emph{single-shard transaction} if $\abs{\Shards{\Transaction}} = 1$ and is a \emph{multi-shard transaction} otherwise. We assume

\begin{assumption}\label{ass:shard_minimality}
Let $D(\Shard)$ be the dataset maintained by shard $\Shard$. We have $\Object \in D(\Shard)$ only if $\ObjectShard{\Object} = \Shard$.
\end{assumption}

\section{Correctness of multi-shard transaction processing}\label{sec:correct}

Before we introduce \CB{}, we put forward the correctness requirements we want to maintain in a multi-shard transaction system in which each shard is itself a set of replicas operated as a Byzantine fault-tolerant system. We say that a shard $\Shard$ performs an action if every good replica in $\NonFaulty{\Shard}$ performs that action. Hence, any processing decision or execution step performed by $\Shard$ requires the usage of a \emph{consensus protocol} to coordinate the replicas in $\Shard$:

\paragraph{Fault-tolerant primitives}
At the core of resilient systems are \emph{consensus protocols}~\cite{pbftj,paxossimple,blockchain_iot,wild} that coordinate the operations of individual replicas in the system, e.g., a Byzantine fault-tolerant system driven by \Pbft{}~\cite{pbftj} or \HS{}~\cite{hotstuff}, or a crash fault-tolerant system driven by \Paxos{}~\cite{paxossimple}. As these systems are fully-replicated, each replica holds exactly the same data, which is determined by the \emph{sequence of transactions}---the journal---agreed upon via consensus:

\begin{definition}\label{def:consensus}
A \emph{consensus protocol} coordinate decision making among the replicas of a resilient cluster $\Shard$ by providing a reliable ordered replication of \emph{decisions}. To do so, consensus protocols provide the following guarantees:
\begin{enumerate}
\item If good replica $\Replica \in \Shard$ makes a $\rn$-th decision, then all good replicas $\Replica' \in \Shard$ will make a $\rn$-th decision (whenever communication becomes reliable).
\item If good replicas $\Replica, \Replica[q] \in \Shard$ make $\rn$-th decisions, then they make the same decisions.
\item Whenever a good replica learns that a decision $D$ needs to be made, then it can force consensus on $D$.
\end{enumerate}
\end{definition}

Let $\Transaction$ be a transaction processed by a sharded fault-tolerant system. Processing of $\Transaction$ does not imply execution: the transaction could be invalid (e.g., the owners of affected objects did not express their support) or the transaction could have inputs that no longer exists. We say that the system \emph{commits} to  $\Transaction$ if it decides to apply the modifications prescribed by $\Transaction$, and we say that the system \emph{aborts} $\Transaction$ if it decides to not do so. Using this terminology, we put forward the following requirements for any sharded fault-tolerant system:
\begin{requirement}
\item\label{req:validity} \emph{Validity}. The system must only processes transaction $\Transaction$ if, for every input object $\Object \in \Inputs{\Transaction}$ with a well-behaved owner $\Owner{\Object}$, the owner $\Owner{\Object}$ supports the transaction.
\item\label{req:shard_inv} \emph{Shard-involvement}. The shard $\Shard$ only processes transaction $\Transaction$ if $\Shard \in \Shards{\Transaction}$.
\item\label{req:shard_app} \emph{Shard-applicability}. Let $D(\Shard)$ be the dataset maintained by shard $\Shard$ at time $t$. The shards $\Shards{\Transaction}$ only commit to execution of transaction $\Transaction$ at $t$ if $\Transaction$ consumes only existing objects. Hence, $\Inputs{\Transaction} \subseteq \bigcup \{ D(\Shard) \mid \Shard \in \Shards{\Transaction} \}$.
\item\label{req:shard_const} \emph{Cross-shard-consistency}. If shard $\Shard$ commits (aborts) transaction $\Transaction$, then all shards $\Shard' \in \Shards{\Transaction}$ eventually commit (abort) $\Transaction$. 
\item\label{req:shard_service} \emph{Service}. If client $\Client$ is well-behaved and wants to request a valid transaction $\Transaction$, then the sharded system will eventually \emph{process} $\Request{\Transaction}{\Client}$. If $\Transaction$ is shard-applicable, then the sharded system will eventually \emph{execute} $\Request{\Transaction}{\Client}$.
\item\label{req:shard_confirm} \emph{Confirmation}. If the system processes $\Request{\Transaction}{\Client}$ and $\Client$ is well-behaved, then $\Client$ will eventually learn whether $\Transaction$ is committed or aborted.
\end{requirement} 
We notice that shard-involvement is a \emph{local requirement}, as individual shards can determine whether they need to process a given transaction. In the same sense, shard-applicability and cross-shard-consistency are \emph{global} requirements, as assuring these requirements requires coordination between the shards affected by a transaction.

\section{Core-\CB{}: simple yet efficient transaction processing}\label{sec:ccb}

The core idea of \CB{} is to minimize the coordination necessary for multi-shard ordering and execution of transactions. To do so, \CB{} combines the semantics of transactions in the object-dataset model with the minimal coordination required to assure shard-applicability and cross-shard consistency. This combination results in the following high-level three-step approach towards processing any transaction $\Transaction$:

\begin{enumerate}
    \item \emph{Local inputs}. First, every affected shard $\Shard \in \Shards{\Transaction}$ locally determines whether it has all inputs from $\Shard$ that are necessary to process $\Transaction$.
    \item \emph{Cross-shard exchange}. Then, every affected shard $\Shard \in \Shards{\Transaction}$ exchanges these inputs to all other shards in $\Shards{\Transaction}$, thereby pledging to use their local inputs in the execution of $\Transaction$.
    \item \emph{Decide outcome} Finally, every affected shard $\Shard \in \Shards{\Transaction}$ decides to commit $\Transaction$ if all affected shards were able to provide all local inputs necessary for execution of $\Transaction$.
\end{enumerate}

Next, we describe how these three high-level steps are incorporated by $\CB$ into normal consensus steps at each shards. Let shard $\Shard \in \Shards{\Replicas}$ receive client request $\Request{\Transaction}{\Client}$. The good replicas in $\Shard$ will first determine whether $\Transaction$ is valid and applicable. If $\Transaction$ is not valid or $\Shard \notin \Shards{\Transaction}$, then the good replicas discard $\Transaction$. Otherwise, if $\Transaction$ is valid and $\Shard \in \Shards{\Transaction}$, then the good replicas utilize \emph{consensus} to force the primary $\Primary{\Shard}$ to propose in some consensus round $\rn$ the message $m(\Shard, \Transaction)_{\rn} =  (\Request{\Transaction}{\Client}, I(\Shard, \Transaction), D(\Shard, \Transaction))$, in which $I(\Shard, \Transaction) = \{ \Object \in \Inputs{\Transaction} \mid \Shard = \ObjectShard{\Object} \}$ is the set of objects maintained by $\Shard$ that are input to $\Transaction$ and $D(\Shard, \Transaction) \subseteq I(\Shard, \Transaction)$ is the set of currently-available inputs at $\Shard$. Only if  $I(\Shard, \Transaction) = D(\Shard, \Transaction)$ will $\Shard$ pledge to use the local inputs $I(\Shard, \Transaction)$ in the execution of $\Transaction$.

The acceptance of $m(\Shard, \Transaction)_{\rn}$ in round $\rn$ by all good replicas completes the \emph{local inputs} step. Next, during execution of $\Transaction$, the \emph{cross-shard exchange} and \emph{decide outcome} steps are performed. First, the \emph{cross-shard exchange} step. In this step, $\Shard$ broadcasts $m(\Shard, \Transaction)_{\rn}$ to all other shards in $\Shards{\Transaction}$. To assure that the broadcast arrives, we rely on a reliable primitive for \emph{cross-shard exchange}, e.g., via an efficient cluster-sending protocol~\cite{disc_csp,geobft}. Then, the replicas in $\Shard$ wait until they receive messages $m(\Shard', \Transaction)_{\rn'} = (\Request{\Transaction}{\Client}, I(\Shard', \Transaction), D(\Shard', \Transaction))$ from all other shards $\Shard' \in \Shards{\Transaction}$.

After cross-shard exchange comes the final \emph{decide outcome} step. After $\Shard$ receives $m(\Shard', \Transaction)_{\rn'}$ from all shards $\Shard' \in \Shards{\Transaction}$, it decides to \emph{commit} whenever  $I(\Shard', \Transaction) = D(\Shard', \Transaction)$ for all $\Shard' \in \Shards{\Transaction}$. Otherwise, it decides \emph{abort}. If $\Shard$ decides commit, then all good replicas in $\Shard$ destruct all objects in $D(\Shard, \Transaction)$ and construct all objects $\Object \in \Outputs{\Transaction}$ with $\Shard = \ObjectShard{\Object}$. Finally, each good replica informs $\Client$ of the outcome of execution. If $\Client$ receives, from every shard $\Shard'' \in \Shards{\Transaction}$, identical outcomes from $\nf{\Shard''} - \f{\Shard''}$ distinct replicas in $\Shard''$, then it considers $\Transaction$ to be successfully executed. In Figure~\ref{fig:flow_ccb}, we sketched the working of \cCB{}. 

\begin{figure}[t!]
    \centering
    \scalebox{0.825}{
    \begin{tikzpicture}[yscale=0.625,xscale=1.45,every edge/.append style={thick},label/.append style={below=5pt,align=center,font=\footnotesize}]
        \node[left] (c) at  (0.75, 3.5) {$\Client$};
        \node[left] (s1) at (0.75, 2) {$\Shard_1$};
        \node[left] (s2) at (0.75, 1) {$\Shard_2$};
        \node[left] (s3) at (0.75, 0) {$\Shard_3$};
        
        \path (0.75, 0) edge ++(9.45, 0)
              (0.75, 1) edge ++(9.45, 0)
              (0.75, 2) edge ++(9.45, 0)
              (0.75, 3.5) edge[green!70!black] ++(9.45, 0);

        \path (1, 0) edge ++(0, 3.5)
              (2, 0) edge ++(0, 3.5)
              (5, 0) edge ++(0, 3.5)
              (6, 0) edge ++(0, 3.5)
              (9, 0) edge ++(0, 3.5)
              (10, 0) edge ++(0, 3.5);

        \path[->] (1, 3.5) edge node[above] {$\Request{\Transaction}{\Client}$} (2, 2) edge  (2, 1) edge  (2, 0)
                  (5, 2) edge (6, 1) edge (6, 0)
                  (5, 1) edge (6, 2) edge (6, 0)
                  (5, 0) edge (6, 2) edge (6, 1)
                  (9, 0) edge (10, 3.5)
                  (9, 1) edge (10, 3.5)
                  (9, 2) edge (10, 3.5);
        
        \draw[draw=orange!60,fill=orange!40,rounded corners] (2.1, -0.4) rectangle (4.9, 0.4);
        \draw[draw=orange!60,fill=orange!40,rounded corners] (2.1,  0.6) rectangle (4.9, 1.4);
        \draw[draw=orange!60,fill=orange!40,rounded corners] (2.1,  1.6) rectangle (4.9, 2.4);
        \node at (3.5, 0) {Consensus on $\Request{\Transaction}{\Client}$};
        \node at (3.5, 1) {Consensus on $\Request{\Transaction}{\Client}$};
        \node at (3.5, 2) {Consensus on $\Request{\Transaction}{\Client}$};

        \draw[draw=yellow!60,fill=yellow!40,rounded corners] (6.1, -0.4) rectangle (8.9, 0.4);
        \draw[draw=yellow!60,fill=yellow!40,rounded corners] (6.1,  0.6) rectangle (8.9, 1.4);
        \draw[draw=yellow!60,fill=yellow!40,rounded corners] (6.1,  1.6) rectangle (8.9, 2.4);
        \node at (7.5, 0) {Wait for Commit/Abort};
        \node at (7.5, 1) {Wait for Commit/Abort};
        \node at (7.5, 2) {Wait for Commit/Abort};

        \node[label] at (3.5, 0) {Local Inputs\\(Consensus)};
        \node[label] at (5.5, 0) {Cross-Shard Exchange\\(Cluster-Sending)};
        \node[label] at (7.5, 0) {Decide Outcome};
        \node[label] at (9.5, 0) {Inform};
    \end{tikzpicture}
    }
    \caption{The message flow of \cCB{} for a $3$-shard client request $\Request{\Transaction}{\Client}$ that is committed.}\label{fig:flow_ccb}
\end{figure}

The \emph{cross-shard exchange} step of \cCB{} at $\Shard$ involves waiting for other shards $\Shard'$. Hence, there is the danger of deadlocks if the other shards $\Shard'$ never perform the cross-shard exchange step:

\begin{example}\label{ex:ccb_concurrent}
Consider distinct transactions $\Request{\Transaction_1}{\Client_1}$ and $\Request{\Transaction_2}{\Client_2}$ that both affect objects $\Object_1$ in $\Shard_1$ and $\Object_2$ in $\Shard_2$. Hence, we have $\Inputs{\Transaction_1} = \Inputs{\Transaction_2} = \{ \Object_1, \Object_2 \}$ and with $\ObjectShard{\Object_1} = \Shard_1$ and $\ObjectShard{\Object_2} = \Shard_2$. We assume that $\Shard_1$ processes $\Transaction_1$ first and $\Shard_2$ processes $\Transaction_2$ first. Shard $\Shard_1$ will start by sending $m(\Shard, \Transaction_1)_{\rn_1} = (\Request{\Transaction_1}{\Client_1}, \{ \Object_1 \}, \{ \Object_1 \})$ to $\Shard_2$. Next, $\Shard_1$ will wait, during which it receives $\Transaction_2$. At the same time, $\Shard_2$ follows similar steps for $\Transaction_2$ and sends $m(\Shard, \Transaction_2)_{\rn_2} = (\Request{\Transaction_2}{\Client_2}, \{ \Object_2 \}, \{ \Object_2 \})$ to $\Shard_1$. As a result, $\Shard_1$ is waiting for information on $\Transaction_1$ from $\Shard_2$, while $\Shard_2$ is waiting for information on $\Transaction_2$ from $\Shard_1$.
\end{example}

To assure that the above example does not lead to a deadlock, we employ two techniques.
\begin{enumerate}
\item \emph{Internal propagation}. To deal with situations in which some shards $\Shard \in \Shards{\Transaction}$ did not receive $\Request{\Transaction}{\Client}$ (e.g., due to network failure or due to a faulty client that fails to send $\Request{\Transaction}{\Client}$ to some shards), we allow each shard to learn $\Transaction$ from any other shard. In specific, any shard $\Shard \in \Shards{\Transaction}$ will start consensus on $\Request{\Transaction}{\Client}$ after receiving \emph{cross-shard exchange} related to $\Request{\Transaction}{\Client}$.

\item \emph{Concurrent resolution}. To deal with concurrent transactions, as in Example~\ref{ex:ccb_concurrent}, we allow each shard to accept and execute transactions for different rounds concurrently. To assure that such concurrent execution does not lead to inconsistent state updates, each replica implements the following \emph{first-pledge} and \emph{ordered-commit} rules. Let $\Transaction$ be a transaction with $\Shard \in \Shards{\Transaction}$ and $\Replica \in \Shard$. The \emph{first-pledge} rule states that $\Shard$ pledges $\Object$, constructed in round $\rn$, to transaction $\Transaction$ only if $\Transaction$ is the first transaction proposed after round $\rn$ with $\Object \in \Inputs{\Transaction}$. The \emph{ordered-commit} rule states that $\Shard$ can abort $\Transaction$ in any order, but will only commit $\Transaction$ that is accepted in round $\rn$ after previous rounds finished execution.
\end{enumerate}

Next, we apply the above two techniques to the situation of Example~\ref{ex:ccb_concurrent}:
\begin{example}
While $\Shard_1$ is waiting for $\Transaction_1$, it received cross-shard exchange related to $\Request{\Transaction_2}{\Client_2}$. Hence, in a future round $\rn_1' > \rn_1$, it can propose and accept $\Request{\Transaction_2}{\Client_2}$. By the first-pledge rule, $\Shard_1$ already pledged $\Object_1$ to the execution of $\Transaction_1$. Hence, it cannot pledge any objects to the execution of $\Transaction_2$. Consequently, $\Shard_1$ will eventually be able to send $m(\Shard_1, \Transaction_2)_{\rn_1'} = (\Request{\Transaction_2}{\Client_2}, \{ \Object_1 \}, \emptyset)$ to $\Shard_2$. Likewise, $\Shard_2$ will eventually be able to send $m(\Shard_2, \Transaction_1)_{\rn_2'} = (\Request{\Transaction_1}{\Client_1}, \{ \Object_2 \}, \emptyset)$ to $\Shard_1$. Consequently, both shards decide abort on $\Transaction_1$ and $\Transaction_2$, which they can do in any order due to the ordered-commit rule.
\end{example}

Finally, we notice that abort decisions at shard $\Shard$ on a transaction $\Transaction$ can often be made without waiting for all shards $\Shard' \in \Shards{\Transaction}$. Shard $\Shard$ can decide abort after it detects $I(\Shard, \Transaction) \neq D(\Shard, \Transaction)$ or after it receives the first message $(\Request{\Transaction}{\Client}, I(\Shard'', \Transaction), D(\Shard'', \Transaction))$ with $I(\Shard'', \Transaction) \neq D(\Shard'', \Transaction)$, $\Shard''\in \Shards{\Transaction}$. For efficiency, we allow $\Shard$ to abort in these cases.

\begin{theorem}\label{thm:ccb}
If, for all shards $\Shard^{\ast}$, $\nf{\Shard^{\ast}} > 2\f{\Shard^{\ast}}$, and Assumptions~\ref{ass:coordinate}, \ref{ass:inout}, \ref{ass:wellclient}, and \ref{ass:shard_minimality} hold, then Core-\CB{} satisfies Requirements~\ref{req:validity}--\ref{req:shard_confirm} with respect to all transactions that are not requested by malicious clients and that do not involve objects with malicious owners.
\end{theorem}
\begin{proof}
Let $\Transaction$ be a transaction. As good replicas in $\Shard$ discard $\Transaction$ if it is invalid or if $\Shard \notin \Shards{\Transaction}$,  \cCB{} provides \emph{validity} and \emph{shard-involvement}. Next, \emph{shard-applicability} follow directly from the decide outcome step.

If a shard $\Shard$ commits or aborts transaction $\Transaction$, then it must have completed the decide outcome and cross-shard exchange steps. As $\Shard$ completed cross-shard exchange, all shards $\Shard' \in \Shards{\Transaction}$ must have exchanged the  necessary information to $\Shard$. By relying on cluster-sending for cross-shard exchange, $\Shard'$ requires cooperation of all good replicas in $\Shard'$ to exchange the necessary information to $\Shard$. Hence, we have the guarantee that these good replicas will also perform cross-shard exchange to any other shard $\Shard'' \in \Shards{\Transaction}$. As such, every shard $\Shard'' \in \Shards{\Transaction}$ will receive the same information as $\Shard$, complete cross-shard exchange, and make the same decision during the decide outcome step, providing \emph{cross-shard consistency}.

Finally, due to internal propagation and concurrent resolution, every valid transaction $\Transaction$ will be processed by \cCB{} as soon as it is send to any shard $\Shard \in \Shards{\Transaction}$. Hence, every shard in $\Shards{\Transaction}$ will perform the necessary steps to eventually inform the client. As all good replicas $\Replica \in \Shard$, $\Shard \in \Shards{\Transaction}$, will inform the client of the outcome for $\Transaction$, the majority of these inform-messages come from good replicas, enabling the client to reliably derive the true outcome. Hence, \cCB{} provides \emph{service} and \emph{confirmation}.
\end{proof}

Notice that in the object-dataset model in which we operate, each object can be constructed once and destructed once. Hence, each object $\Object$ can be part of at-most two committed transactions: the first of which will construct $\Object$ as an output, and the second of which has $\Object$ as an input and will consume and destruct $\Object$. This is independent of any other operations on other objects.  As such these two transactions \emph{cannot} happen concurrently. Consequently, we only have concurrent transactions on $\Object$ if the owner $\Owner{\Object}$ expresses support for several transactions that have $\Object$ as an input. By Assumption~\ref{ass:wellclient}, the owner $\Owner{\Object}$ must be malicious in that case. As such, transactions of well-behaved clients and owners will \emph{never abort}.

In the design of \cCB{}, we take advantage of this observation that \emph{aborts} are always due to malicious behavior by clients and owners of objects: to minimize coordination while allowing graceful resolution of concurrent transactions, we do not undo any pledges of objects to the execution of any transactions. This implies that objects that are involved in malicious transactions can get lost for future usage, while not affecting any transactions of well-behaved clients.

\section{Optimistic-\CB{}: robust transaction processing}\label{sec:ocb}
In the previous section, we introduced \cCB{}, a minimalistic and efficient multi-shard transaction processing protocol that relies on practical properties of UTXO-like transactions. Although the design of \cCB{} is simple yet effective, we see two shortcomings that limits its use cases. First, \cCB{} operates under the assumption that any issues arising from concurrent transactions is due to malicious behavior of clients. As such, \cCB{} chooses to lock out objects affected by such malicious behavior for any future usage. Second, \cCB{} requires consecutive consensus and cluster-sending steps, which increases its transaction processing latencies. Next, we investigate how to deal with these weaknesses of \cCB{} \emph{without giving up} on the minimalistic nature of \cCB{}.

To do so, we propose Optimistic-\CB{} (\oCB{}), which is optimized for the \emph{optimistic} case in which we have no concurrent transactions, while providing a recovery path that can recover from concurrent transactions without locking out objects. At the core of \oCB{} is assuring that any issues due to malicious behavior, e.g., concurrent transactions, are \emph{detected} in such a way that individual replicas can start a recovery process. At the same time, we want to minimize transaction processing latencies. To bridge between these two objectives, we integrate detection and cross-shard coordination within a single consensus round that runs at each affected shard.

Let $\Request{\Transaction}{\Client}$ be a multi-shard transaction, let $\Shard \in \Shards{\Transaction}$ be a shard with primary $\Primary{\Shard}$, and let $m(\Shard, \Transaction)_{\View, \rn}$ be the round-$\rn$ proposal of $\Primary{\Shard}$ of view $\View$ of $\Shard$. To enable detection of concurrent transactions, \oCB{} modifies the consensus-steps of the underlying consensus protocol by applying the following high-level idea:

\begin{quote}
A replica $\Replica \in \Shard$, $\Shard \in \Shards{\Transaction}$, only accepts proposal $m(\Shard, \Transaction)_{\View, \rn} = (\Request{\Transaction}{\Client}, I(\Shard, \Transaction), D(\Shard, \Transaction))$ for some transaction $\Transaction$ if it gets confirmation that replicas in each other shard $\Shard' \in \Shards{\Transaction}$ are also accepting proposals for $\Transaction$. Otherwise, replica $\Replica$ detects failure.
\end{quote}

Next, we illustrate how to integrate the above idea in the three-phase design of \Pbft{}, thereby turning \Pbft{} into a multi-shard aware consensus protocol:
\begin{enumerate}
    \item \emph{Global preprepare}. Primary $\Primary{\Shard}$ must send $m(\Shard, \Transaction)_{\View, \rn}$ to all replicas $\Replica' \in \Shard'$, $\Shard' \in \Shards{\Transaction}$. Replica $\Replica \in \Shard$ only finishes the global preprepare phase after it receives a \emph{global preprepare certificate} consisting of a set $M = \{ m(\Shard'', \Transaction)_{\View'', \rn''} \mid \Shard'' \in \Shards{\Transaction} \}$ of preprepare messages from all primaries of shards affected by $\Transaction$.
    \item \emph{Global prepare}. After $\Replica \in \Shard$, $\Shard \in \Shards{\Transaction}$, finishes the global preprepare phase, it sends prepare messages for $M$ to all other replicas in $\Replica' \in \Shard'$, $\Shard' \in \Shards{\Transaction}$. Replica $\Replica \in \Shard$ only finishes the global prepare phase for $M$ after, for every shard $\Shard' \in \Shards{\Transaction}$, it receives a \emph{local prepare certificate} consisting of a set $P(\Shard')$ of prepare messages for $M$ from $\nf{\Shard'}$ distinct replicas in $\Shard'$. We call the set $\{ P(\Shard'') \mid \Shard'' \in \Shards{\Transaction} \}$ a \emph{global prepare certificate}.
    \item \emph{Global commit}. After replica $\Replica \in \Shard$, $\Shard \in \Shards{\Transaction}$, finishes the global prepare phase, it sends commit messages for $M$ to all other replicas in $\Replica' \in \Shard'$, $\Shard' \in \Shards{\Transaction}$. Replica $\Replica \in \Shard$ only finishes the global commit phase for $M$ after, for every shard $\Shard' \in \Shards{\Transaction}$, it receives a \emph{local commit certificate} consisting of a set $C(\Shard')$ of commit messages for $M$ from $\nf{\Shard'}$ distinct replicas in $\Shard'$. We call the set $\{ P(\Shard'') \mid \Shard'' \in \Shards{\Transaction} \}$ a \emph{global commit certificate}.
\end{enumerate}

To minimize inter-shard communication, one can utilize threshold signatures and cluster-sending to carry over local prepare and commit certificates between shards via a few constant-sized messages. The above three-phase \emph{global-\Pbft{}} protocol already takes care of the \emph{local input} and \emph{cross-shard exchange} steps. Indeed, a replica $\Replica \in \Shard$ that finishes the global commit phase has accepted global preprepare certificate $M$, which contains all information of other shards to proceed with execution. At the same time, $\Replica$ also has confirmation that $M$ is prepared by a majority of all good replicas in each shard $\Shard' \in \Shards{\Transaction}$ (which will eventually be followed by execution of $\Transaction$ within $\Shard'$). With these ingredients in place, only the \emph{decide outcome} step remains.

The decide outcome step at shard $\Shard$ is entirely determined by the global preprepare certificate $M$. Shard $\Shard$ decides to \emph{commit} whenever $I(\Shard', \Transaction) = D(\Shard', \Transaction)$ for all $(\Request{\Transaction}{\Client}, I(\Shard', \Transaction), D(\Shard', \Transaction)) \in M$. Otherwise, it decides \emph{abort}. If $\Shard$ decides commit, then all good replicas in $\Shard$ destruct all objects in $D(\Shard, \Transaction)$ and construct all objects $\Object \in \Outputs{\Transaction}$ with $\Shard = \ObjectShard{\Object}$. Finally, each good replica informs $\Client$ of the outcome of execution.  If $\Client$ receives, from every shard $\Shard' \in \Shards{\Transaction}$, identical outcomes from $\nf{\Shard'} - \f{\Shard}$ distinct replicas in $\Shard'$, then it considers $\Transaction$ to be successfully executed. In Figure~\ref{fig:flow_ocb}, we sketched the working of \oCB{}. 

\begin{figure}[t!]
    \centering
    \scalebox{0.825}{
    \begin{tikzpicture}[yscale=0.625,xscale=1.45,every edge/.append style={thick},label/.append style={below=5pt,align=center,font=\footnotesize}]
        \node[left] (c) at  (0.75, 3.5) {$\Client$};
        \node[left] (s1) at (0.75, 2) {$\Shard_1$};
        \node[left] (s2) at (0.75, 1) {$\Shard_2$};
        \node[left] (s3) at (0.75, 0) {$\Shard_3$};
        
        \draw[draw=orange!60,fill=orange!40,rounded corners] (2.1, -0.4) rectangle (5.9, 2.4);
        \path (0.75, 0) edge ++(9.45, 0)
              (0.75, 1) edge ++(9.45, 0)
              (0.75, 2) edge ++(9.45, 0)
              (0.75, 3.5) edge[green!70!black] ++(9.45, 0);

        \path (1, 0) edge ++(0, 3.5)
              (2, 0) edge ++(0, 3.5)
              (6, 0) edge ++(0, 3.5)
              (9, 0) edge ++(0, 3.5)
              (10, 0) edge ++(0, 3.5);

        \path[->] (1, 3.5) edge node[above] {$\Request{\Transaction}{\Client}$} (2, 2) edge  (2, 1) edge  (2, 0)
                  (9, 0) edge (10, 3.5)
                  (9, 1) edge (10, 3.5)
                  (9, 2) edge (10, 3.5);
        
        \path[->] (2.5, 2) edge (3.5, 1) edge (3.5, 0)
                  (2.5, 1) edge (3.5, 2) edge (3.5, 0)
                  (2.5, 0) edge (3.5, 1) edge (3.5, 2)
                  (3.5, 2) edge (4.5, 1) edge (4.5, 0)
                  (3.5, 1) edge (4.5, 2) edge (4.5, 0)
                  (3.5, 0) edge (4.5, 1) edge (4.5, 2)
                  (4.5, 2) edge (5.5, 1) edge (5.5, 0)
                  (4.5, 1) edge (5.5, 2) edge (5.5, 0)
                  (4.5, 0) edge (5.5, 1) edge (5.5, 2);

        \node[below, label] at (3, 0) {Preprepare};
        \node[below, label] at (4, 0) {Prepare};
        \node[below, label] at (5, 0) {Commit};
        \draw[decoration={brace, mirror},decorate,thick] (2.1, -1) -- (5.9, -1);
        \node[label] at (4, -1) {Local Inputs and Cross-Shard Exchange\\(Global Consensus)};
        \draw[draw=yellow!60,fill=yellow!40,rounded corners] (6.1, -0.4) rectangle (8.9, 0.4);
        \draw[draw=yellow!60,fill=yellow!40,rounded corners] (6.1,  0.6) rectangle (8.9, 1.4);
        \draw[draw=yellow!60,fill=yellow!40,rounded corners] (6.1,  1.6) rectangle (8.9, 2.4);
        \node at (7.5, 0) {Decide for Commit/Abort};
        \node at (7.5, 1) {Decide for Commit/Abort};
        \node at (7.5, 2) {Decide for Commit/Abort};

        \node[label] at (7.5, 0) {Decide Outcome};
        \node[label] at (9.5, 0) {Inform};
    \end{tikzpicture}
    }
    \caption{The message flow of \oCB{} for a $3$-shard client request $\Request{\Transaction}{\Client}$ that is committed.}\label{fig:flow_ocb}
\end{figure}

Due to the similarity between \oCB{} and \cCB{}, it is straightforward to use the details of Theorem~\ref{thm:ccb} to prove that \oCB{} provides \emph{validity}, \emph{shard-involvement}, and \emph{shard-applicability}.   Next, we will focus on how \oCB{} provides \emph{cross-shard-consistency}. As a first step, we illustrate the ways in which the normal-case of \oCB{}  can fail (e.g., due to malicious behavior of clients, failing replicas, or unreliable communication).

\begin{example}\label{ex:ocb_failures}
Consider a transaction $\Transaction$ proposed by client $\Client$ and affecting shard $\Shard \in \Shards{\Transaction}$. First, we consider the case in which $\Primary{\Shard}$ is malicious and tries to set up a coordinated attack. To have maximum control over the steps of \oCB{}, the primary sends  the message $m(\Shard, \Transaction)_{\View, \rn}$ to only $\nf{\Shard''} - \f{\Shard''}$ good replicas in each shard $\Shard'' \in \Shards{\Transaction}$. By doing so, $\Primary{\Shard}$ can coordinate the faulty replicas in each shard to assure failure of any phase at any replica $\Replica' \in \Shard'$, $\Shard' \in \Transaction$: 
\begin{enumerate}
    \item To prevent $\Replica'$ from finishing the global preprepare phase (and start the global prepare phase) for an $M$ with $m(\Shard',\Transaction)_{\View',\rn'} \in M$, $\Primary{\Shard}$ simply does not send $m(\Shard, \Transaction)_{\View, \rn}$ to $\Replica'$.
    \item To prevent $\Replica'$ from finishing the global prepare phase (and start the global commit phase) for $M$, $\Primary{\Shard}$ instructs the faulty replicas in $\Faulty{\Shard}$ to not send prepare messages for $M$ to $\Replica'$. Hence, $\Replica'$ will receive at-most $\nf{\Shard} - \f{\Shard}$ prepare messages for $M$ from replicas in $\Shard$, assuring that it will not receive a local prepare certificate $P(\Shard)$ and will not finish the global prepare phase for $M$.
\item Likewise, to prevent $\Replica'$ from finishing the global commit phase (and start execution) for $M$, $\Primary{\Shard}$ instructs the faulty replicas in $\Faulty{\Shard}$ to not send commit messages to $\Replica'$.  Hence, $\Replica'$ will receive at-most $\nf{\Shard} - \f{\Shard}$ commit messages for $M$ from replicas in $\Shard$, assuring that it will not receive a local commit certificate $C(\Shard)$ and will not finish the global commit phase for $M$.
\end{enumerate}

None of the above attacks can be attributed to faulty behavior of $\Primary{\Shard}$. First, unreliable communication can result in the same outcomes for $\Replica'$. Furthermore, even if communication is reliable and $\Primary{\Shard}$ is good, we can see the same outcomes:
\begin{enumerate}
\item The client $\Client$ can be malicious and not send $\Transaction$ to $\Shard$. At the same time, all other primaries $\Primary{\Shard''}$ of shards $\Shard'' \in \Shards{\Transaction}$ can be malicious and not send anything to $\Shard$ either. In this case, $\Primary{\Shard}$ will never be able to send any message $m(\Shard, \Transaction)_{\View, \rn}$ to $\Replica'$, as no replica in $\Shard$ is aware of $\Transaction$.
\item If any primary $\Primary{\Shard''}$ of $\Shard'' \in \Shards{\Transaction}$ is malicious, then it can prevent some replicas in $\Shard$ from starting the global prepare phase, thereby preventing these replicas to send prepare messages to $\Replica'$. If $\Primary{\Shard''}$ prevents sufficient replicas in $\Shard$ from starting the global prepare phase, $\Replica'$ will be unable to finish the global prepare phase.
\item Likewise, any malicious primary  $\Primary{\Shard''}$ of $\Shard'' \in \Shards{\Transaction}$  can prevent replicas in $\Shard$ from starting the global commit phase, thereby assuring that $\Replica'$ will be unable to finish the global commit phase.
\end{enumerate}
\end{example}

To deal with malicious behavior, \oCB{} needs a robust recovery mechanism. We cannot simply build that mechanism on top of traditional view-change approaches: these traditional view-change approaches require that one can identify a single source of failure (when communication is reliable), namely the current primary. As Example~\ref{ex:ocb_failures} already showed, this property does not hold for \oCB{}. To remedy this, the recovery mechanisms of \oCB{} has components that perform \emph{local view-change} and that perform \emph{global state recovery}. The pseudo-code for this recovery protocol can be found in Figure~\ref{fig:recovery_main}. Next, we describe the working of this recovery protocol in detail. Let $\Replica \in \Shard$ be a replica that determines that it cannot finish a round $\rn$ of view $\View$.

First, $\Replica$ determines whether it already has a \emph{guarantee} on which transaction it has to execute in round $\rn$. This is the case when the following conditions are met: $\Replica$ finished the global prepare phase for $M$ with $m(\Shard, \Transaction)_{\View,\rn} \in M$ and has received a local commit certificate $C(\Shard'')$ for $M$ from some shard $\Shard'' \in \Shards{\Transaction}$. In this case, $\Replica$ can simply request all missing local commit certificates directly, as $C(\Shard'')$ can be used to prove to any involved replica $\Replica' \in \Shard'$, $\Shard' \in \Shards{\Transaction}$, that $\Replica'$ also needs to commit to $M$. To request such missing commit certificates of $\Shard'$, replica $\Replica$ sends out \MName{VCGlobalSCR} messages to all replicas in $\Shard'$ (\lfref{fig:recovery_main}{global_send}). Any replica $\Replica'$ that receives such a \Name{VCGlobalSCR} message can use the information in that message to reach the global commit phase for $M$ and, hence, provide $\Replica$ with the requested commit messages (\lfref{fig:recovery_main}{reply}).

\begin{figure}[t!]
\begin{myprotocol}
    \EVENT{$\Replica \in \Shard$ is unable to finish round $\rn$ of view $\View$}\label{fig:recovery_main:detect_event}
        \IF{$\Replica$ finished in round $\rn$ the global prepare phase for $M$,
                \\\qquad but is unable to finish the global commit phase}
        \STATE Let $P$ be the global prepare certificate of $\Replica$ for $M$.
        \IF{$\Replica$ has a local commit certificate $C(\Shard'')$ for $M$}\label{fig:recovery_main:force_commit}
            \FOR{$\Shard' \in \Shards{\Transaction}$}
                \IF{$\Replica$ did not yet receive a local commit certificate $C(\Shard')$}
                    \STATE Broadcast $\Message{VCGlobalSCR}{M, P, C(\Shard'')}$ to all replicas in $\Shard'$.\label{fig:recovery_main:global_send}
                \ENDIF
            \ENDFOR
        \ELSE
            \STATE Detect the need for local state recovery of round $\rn$ of view $\View$ (Figure~\ref{fig:local_recovery}).
        \ENDIF
    \ELSE
        \STATE Detect the need for local state recovery of round $\rn$ of view $\View$ (Figure~\ref{fig:local_recovery}).
    \ENDIF
    \STATE (Eventually repeat this event if $\Replica$ remains unable to finish round $\rn$.)
    \ENDEVENT
    \SPACE
    \EVENT{$\Replica' \in \Shard'$ receives a message $\Message{VCGlobalSCR}{M, P, C(\Shard'')}$ from $\Replica \in \Shard$}\label{fig:recovery_main:reply}
        \IF{$\Replica'$ did not reach the global commit phase for $M$}
            \STATE Use $M$, $P$, and $C(\Shard'')$ to reach the global commit phase for $M$.
        \ELSE
            \STATE Send a commit message for $M$ to $\Replica$.
        \ENDIF
    \ENDEVENT
\end{myprotocol}
\caption{The view-change \emph{global short-cut recovery path} that determines whether $\Replica$ already has the assurance that the current transaction will be committed. If this is the case, then $\Replica$ requests only the missing information to proceed with execution. Otherwise, $\Replica$ requires at-least local recovery (Figure~\ref{fig:local_recovery}).}\label{fig:recovery_main}
\end{figure}

If $\Replica$ does not have a \emph{guarantee} itself on which transaction it has to execute in round $\rn$, then it needs to determine whether any other replica (either in its own shard or in any other shard) has already received and acted upon such a guarantee. To initiate such local and global state recovery, $\Replica$ simply detects the current view as faulty. To do so, $\Replica$ broadcasts a \MName{VCRecoveryRQ} message to all other replicas in $\Shard$ that contains all information $\Replica$ collected on round $\rn$ in view $\View$ (\lfref{fig:local_recovery}{request}). Other replicas $\Replica[q] \in \Shard$ that already have \emph{guarantees} for round $\rn$ can help $\Replica$ by providing all missing information (\lfref{fig:local_recovery}{shortcut}). On receipt of this information, $\Replica$ can proceed with the round (\lfref{fig:local_recovery}{shortcut_reply}). If no replicas can provide the missing information, then eventually all good replicas will detect the need for local recovery, this either by themselves (\lfref{fig:local_recovery}{owndetect}) or after receiving \MName{VCRecoveryRQ} messages of at-least $\f{\Shard}+1$ distinct replicas in $\Shard$, of which at-least a single replica must be good (\lfref{fig:local_recovery}{multi}).

Finally, if a replica $\Replica$ receives $\nf{\Shard}$ \MName{VCRecoveryRQ} messages, then it has the guarantee that at least $\nf{\Shard}-\f{\Shard} \geq \f{\Shard}+1$ of these messages come from good replicas in $\Shard$. Hence, due to \lfref{fig:local_recovery}{multi}, all $\nf{\Shard}$ good replicas in $\Shard$ will send \MName{VCRecoveryRQ}, and, when communication is reliable, also receive these messages. Consequently, at this point, $\Replica$ can start the new view by electing a new primary and awaiting the \MName{NewView} proposal of this new primary (\lfref{fig:local_recovery}{new_view}). If $\Replica$ is the new primary, then it starts the new view by proposing a \MName{NewView}. As other shards \emph{could} have already made final decisions depending on local prepare or commit certificates of $\Shard$ for round $\rn$, we need to assure that such certificates are not invalidated. To figure out whether such final decisions have been made, the new primary will query other shards $\Shard'$ for their state whenever the \MName{NewView} message contains global preprepare certificates for transactions $\Transaction$, $\Shard' \in \Shards{\Transaction}$, but not a local commit certificate to \emph{guarantee} execution of $\Transaction$ (\lfref{fig:local_recovery}{new_global_view}).

\begin{figure}[t!]
\begin{myprotocol}
    \EVENT{$\Replica \in \Shard$ detects the need for local state recovery of round $\rn$ of view $\View$}\label{fig:local_recovery:owndetect}
        \STATE Let $M$ be the latest global preprepare certificate accepted for round $\rn$ by $\Replica$ (if any).\label{fig:local_recovery:provide}
        \STATE Let $S$ be $M$ and any prepare and commit certificates for $M$ collected by $\Replica$.
        \STATE Broadcast $\Message{VCRecoveryRQ}{\View,\rn,S}$.\label{fig:local_recovery:request}
    \ENDEVENT
    \SPACE
    \EVENT{$\Replica[q] \in \Shard$ receives messages $\Message{VCRecoveryRQ}{\View,\rn,S}$ of $\Replica \in \Shard$ and $\Replica[q]$ has
        \begin{enumerate}
            \item started the global prepare phase for $M$ with $m(\Shard, \Transaction)_{\View[w], \rn} \in M$;
            \item a global prepare certificate for $M$;
            \item a local commit certificate $C(\Shard'')$ for $M$
        \end{enumerate}}
        \STATE Send $\Message{VCLocalSCR}{M, P, C(\Shard'')}$ to $\Replica \in \Shard$.\label{fig:local_recovery:shortcut}
    \ENDEVENT
    \SPACE
    \EVENT{$\Replica \in \Shard$ receives a message $\Message{VCLocalSCR}{M, P, C(\Shard'')}$ from $\Replica[q] \in \Shard$}\label{fig:local_recovery:shortcut_reply}
        \IF{$\Replica$ did not reach the global commit phase for $M$}
            \STATE Use $M$, $P$, and $C$ to reach the global commit phase for $M$.
        \ENDIF
    \ENDEVENT
    \SPACE
\EVENT{$\Replica \in \Shard$ receives messages $\Message{VCRecoveryRQ}{\View_i,\rn,S_i}$, $1 \leq i \leq \f{\Shard}+1$,
            \\\qquad from distinct replicas in $\Shard$}\label{fig:local_recovery:multi}
        \STATE $\Replica$ detects the need for local state recovery of round $\rn$ of view $\min\{\View_i \mid 1 \leq i \leq \f{\Shard}+1 \}$.
    \ENDEVENT
    \SPACE
    \EVENT{$\Replica \in \Shard$ receives messages $\Message{VCRecoveryRQ}{\View,\rn,S_i}$, $1 \leq i \leq \nf{\Shard}$,
            \\\qquad from distinct replicas in $\Shard$}\label{fig:local_recovery:new_view}
        \IF{$\ID{\Replica} \neq (\View + 1) \bmod \n{\Shard}$}
            \STATE ($\Replica$ awaits the \MName{NewView} message of the new primary, \lfref{fig:new_view}{process}.)
        \ELSE
            \STATE Broadcast $\Message{NewView}{\Message{VCRecoveryRQ}{\View,\rn,S_i} \mid 1 \leq i \leq \nf{\Shard}}$ to all replicas in $\Shard$.\label{fg:local_recovery:new_view}
            \IF{there exists a $S_i$ that contains global preprepare certificate $M$,\\\qquad but no $S_j$ contains a local commit certificate for $M$}\label{fig:local_recovery:new_global_view}
                \STATE $\Replica$ initiates global state recovery of round $\rn$ (\lfref{fig:new_view}{global}).
            \ENDIF
        \ENDIF
    \ENDEVENT
    \end{myprotocol}
    \caption{The view-change \emph{local short-cut recovery path} that determines whether some $\Replica[q]$ can provide $\Replica$ with the assurance that the current transaction will be committed. If this is the case, then $\Replica$ only needs this assurance, otherwise $\Shard$ requires a new view (Figure~\ref{fig:new_view}).}\label{fig:local_recovery}
\end{figure}

The new-view process has three stages. First, the new primary $\Replica[p]$ proposes the new-view via a \MName{NewView} message (\lfref{fig:local_recovery}{new_view}). If necessary, the new primary $\Replica[p]$ also requests the relevant global state from any relevant shard (\lfref{fig:new_view}{global}). The replicas in other shards will respond to this request with their local state (\lfref{fig:new_view}{other_state}). The new primary collects these responses and sends them to all replicas in $\Shard$ via a \MName{NewViewGlobal} message.

Then, after $\Replica[p]$ sends the \MName{NewView} message to $\Replica \in \Shard$, $\Replica$ determines whether the \MName{NewView} message contains sufficient information to recover round $\rn$ (\lfref{fig:new_view}{recover}), contains sufficient information to wait for any relevant global state (\lfref{fig:new_view}{recover_global_wait}), or to determine that the new primary must propose for round $\rn$ (\lfref{fig:new_view}{replace}). If $\Replica$ determines it needs to wait for any relevant global state, then $\Replica$ will wait for this state to arrive via a \MName{NewViewGlobal} message. Based on the received global state, $\Replica$ determines to  recover round $\rn$ (\lfref{fig:new_view}{global_recover}), or determines that the new primary must propose for round $\rn$ (\lfref{fig:new_view}{global_replace}).

\begin{figure}[t!]
\begin{myprotocol}
    \EVENT{$\Replica[p] \in \Shard$ initiates global state recovery of round $\rn$ using $\Message{NewView}{V}$}\label{fig:new_view:global}
        \STATE Let $T$ be the transactions with global preprepare certificates for round $\rn$ of $\Shard$ in $V$.
        \STATE Let $S$ be the shards affected by transactions in $T$.
        \STATE Broadcast $\Message{VCGlobalStateRQ}{\View, \rn, V}$ to all replicas in $\Shard' \in S$.
        \FOR{$\Shard' \in S$}
            \STATE Wait for \MName{VCGlobalStateRQ} messages for $V$ from $\nf{\Shard'}$ distinct replicas in $\Shard'$.
            \STATE Let $W(\Shard')$ be the set of received \Name{VCGlobalStateRQ} messages.
        \ENDFOR
        \STATE Broadcast $\Message{NewViewGlobal}{V, \{ W(\Shard') \mid \Shard' \in S \}}$ to all replicas in $\Shard$.
    \ENDEVENT
    \SPACE
    \EVENT{$\Replica' \in \Shard'$ receives message $\Message{VCGlobalStateRQ}{\View, \rn, V}$ from $\Replica[p] \in \Shard$}\label{fig:new_view:other_state}
        \IF{$\Replica'$ has a global preprepare certificate $M$ with $m(\Shard, \Transaction)_{\View[w], \rn} \in M$
            \\\qquad and reached the global commit phase for $M$}\label{fig:new_view:global_commit_state}
            \STATE Let $P$ be the global prepare certificate for $M$.
            \STATE Send $\Message{VCGlobalStateR}{\View, \rn, V, M, P}$ to $\Replica[p]$.
        \ELSE
            \STATE Send $\Message{VCGlobalStateR}{\View, \rn, V}$ to $\Replica[p]$.
        \ENDIF
    \ENDEVENT
    \SPACE
    \EVENT{$\Replica \in \Shard$ receives a valid $\Message{NewView}{V}$ message from replica $\Replica[p]$}\label{fig:new_view:process}
        \IF{there exists a $\Message{VCRecoveryRQ}{\View_i,\rn,S_i} \in V$ that contains
            \\\qquad a global preprepare certificate $M$ with $m(\Shard, \Transaction)_{\View[w],\rn} \in M$,
            \\\qquad a global prepare certificate $P$ for $M$, and a local commit certificate $C(\Shard'')$ for $M$}\label{fig:new_view:recover}
            \STATE Use $M$, $P$, and $C$ to reach the global commit phase for $M$.\label{fig:new_view:commit}
        \ELSIF{there exists a $\Message{VCRecoveryRQ}{\View_i,\rn,S_i} \in V$ that contains
            \\\qquad a global preprepare certificate $M$,\\\qquad but no $\Message{VCRecoveryRQ}{\View_j,\rn,S_j} \in V$ contains a local commit certificate for $M$}\label{fig:new_view:recover_global_wait}
            \STATE $\Replica$ detects the need for global state recovery of round $\rn$ (\lfref{fig:new_view}{global_view}).
        \ELSE
            \STATE  ($\Replica[p]$ must propose for round $\rn$.)\label{fig:new_view:replace}
        \ENDIF
    \ENDEVENT
    \SPACE
    \EVENT{$\Replica \in \Shard$ receives a valid $\Message{NewViewGlobal}{V, W}$ from $\Replica[p] \in \Shard$}\label{fig:new_view:global_view}
        \IF{any message in $W$ is of the form $\Message{VCGlobalStateR}{\View, \rn, V, M, P}$}\label{fig:new_view:global_recover}
            \STATE Select $\Message{VCGlobalStateR}{\View, \rn, V, M, P} \in W$ with latest view $\View[w]$, $m(\Shard, \Transaction)_{\View[w],\rn} \in M$.\label{fig:new_view:global_high_select}
            \STATE Use $M$ and $P$ to reach the global commit phase for $M$.\label{fig:new_view:global_commit}
        \ELSE\label{fig:new_view:global_replace}
            \STATE ($\Replica[p]$ must propose for round $\rn$.)
        \ENDIF
    \ENDEVENT
    \end{myprotocol}
\caption{The view-change \emph{new-view recovery path} that recovers the state of the previous view based on a \MName{NewView} proposal of the new primary. As part of the new-view recovery path, the new primary can construct a global new-view that contains the necessary information from other shards to reconstruct the local state.}\label{fig:new_view}
\end{figure}

Next, we shall prove the correctness of the view-change protocol outlined in Figures~\ref{fig:recovery_main},~\ref{fig:local_recovery}, and~\ref{fig:new_view}. First, using a standard quorum argument, we prove that in a single round of a single view of $\Shard$, only a single global preprepare message affecting $\Shard$ can get committed by any other affected shards:

\begin{lemma}\label{lem:non_divergence}
Let $\Transaction_1$ and $\Transaction_2$ be transactions with $\Shard \in (\Shards{\Transaction_1} \intersect \Shards{\Transaction_2})$. If $\nf{\Shard} > 2\f{\Shard}$ and there exists shards $\Shard_i \in \Shards{\Transaction_i}$, $i \in \{1,2\}$, such that good replicas $\Replica_i \in \NonFaulty{\Shard_i}$ reached the global commit phase for global preprepare certificate $M_i$ with $m(\Shard, \Transaction_i)_{\View,\rn} \in M_i$, then $\Transaction_1 = \Transaction_2$.
\end{lemma}
\begin{proof}
We prove this property using contradiction. We assume $\Transaction_1 \neq \Transaction_2$. Let $P_i(\Shard)$ be the local prepare certificate provided by $\Shard$ for $M_i$ and used by $\Replica_i$ to reach the global commit phase, let $S_i \subseteq \Shard$ be the $\nf{\Shard}$ replicas in $\Shard$ that provided the prepare messages in $P_i(\Shard)$, and let $T_i = S_i \difference \Faulty{\Shard}$ be the good replicas in $S_i$. By construction, we have $\abs{T_i} \geq \nf{\Shard} - \f{\Shard}$. As all replicas in $T_1 \union T_2$ are good, they will only send out a single prepare message per round $\rn$ of view $\View$. Hence, if $\Transaction_1 \neq \Transaction_2$, then $T_1 \intersect T_2 = \emptyset$, and we must have $2(\nf{\Shard} - \f{\Shard}) \leq \abs{T_1 \union T_2}$. As all replicas in $T_1 \union T_2$ are good, we also have $\abs{T_1 \union T_2} \leq \nf{\Shard}$. Hence, $2(\nf{\Shard} - \f{\Shard}) \leq \nf{\Shard}$, which simplifies to $\nf{\Shard} \leq 2\f{\Shard}$, a contradiction. Hence, we conclude $\Transaction_1 = \Transaction_2$.
\end{proof}

Next, we use Lemma~\ref{lem:non_divergence} to prove that any global preprepare certificate that \emph{could} have been accepted by any good affected replica is preserved by \oCB{}:

\begin{proposition}\label{prop:view_change_o}
Let $\Transaction$ be a transaction and $m(\Shard, \Transaction)_{\View,\rn}$ be a preprepare message. If, for all shards $\Shard^{\ast}$, $\nf{\Shard^{\ast}} > 2\f{\Shard^{\ast}}$, and there exists a shard $\Shard' \in \Shards{\Transaction}$ such that $\nf{\Shard'} - \f{\Shard'}$ good replicas in $\Shard'$ reached the global commit phase for $M$ with $m(\Shard, \Transaction)_{\View,\rn} \in M$, then every successful future view of $\Shard$ will recover $M$ and assure that the good replicas in $\Shard$ reach the commit phase for $M$.
\end{proposition}
\begin{proof}
Let $\View^{\ast} \leq \View$ be the first view in which a global prepare certificate $M^{\ast}$ with $m(\Shard, \Transaction^{\ast})_{\View^{\ast},\rn} \in M^{\ast}$ satisfied the premise of this proposition. Using induction on the number of views after the first view $\View^{\ast}$, we will prove the following two properties on $M^{\ast}$:
\begin{enumerate}
    \item every good replica that participates in view $\View[w]$, $\View^{\ast} < \View[w]$, will recover $M^{\ast}$ upon entering view $\View[w]$ and reach the commit phase for $M^{\ast}$; and
    \item no replica will be able to construct a local prepare certificate of $\Shard$ for any global preprepare certificate $M^{\dag} \neq M^{\ast}$ with $m(\Shard, \Transaction^{\dag})_{\View[w], \rn} \in M^{\dag}$, $\View^{\ast} < \View[w]$.
\end{enumerate}
The base case is view $\View^{\ast} + 1$. Let $S' \subseteq \NonFaulty{\Shard'}$ be the set of $\nf{\Shard'} - \f{\Shard'}$ good replicas in $\Shard'$ that reached the global commit phase for $M^{\ast}$. Each replica $\Replica' \in S'$ has a local prepare certificate $P(\Shard)$ consisting of $\nf{\Shard}$ prepare messages for $M^{\ast}$ provided by replicas in $\Shard$. We write $S(\Replica') \subseteq \NonFaulty{\Shard}$ to denote the at-least $\nf{\Shard} - \f{\Shard}$ good replicas in $\Shard$ that provided such a prepare message to $\Replica'$.

Consider any valid new-view proposal $\Message{NewView}{V}$ for view $\View^{\ast}+1$. If the conditions of \lfref{fig:new_view}{recover} hold for some global preprepare certificate $M^{\dag}$  with $m(\Shard, \Transaction^{\ddag})_{\View[w], \rn} \in M^{\ddag}$, then we recover $M^{\ddag}$. As there is a local commit certificate for $M^{\ddag}$ in this case, the premise of this proposition holds on $M^{\ddag}$. As $\View^{\ast}$ is the first view in which the premise of this proposition hold, we can use Lemma~\ref{lem:non_divergence} to conclude that $\View[w] = \View^{\ast}$, $M^{\ddag} = M^{\ast}$, and, hence, that the base case holds if the conditions of \lfref{fig:new_view}{recover} hold.  Next, we assume that the conditions of \lfref{fig:new_view}{recover} do not hold, in which case $M^{\ast}$ can only be recovered via global state recovery. As the first step in global state recovery is proving that the condition of \lfref{fig:new_view}{recover_global_wait} holds. Let $T \subseteq \NonFaulty{\Shard}$ be the set of at-least $\nf{\Shard} - \f{\Shard}$ good replicas in $\Shard$ whose \MName{VCRecoveryRQ} message is in $V$ and let $\Replica' \in S'$. We have $\abs{S(\Replica')} \geq \nf{\Shard} -\f{\Shard}$ and $\abs{T} \geq \nf{\Shard} - \f{\Shard}$. Hence, by a standard quorum argument, we conclude $S(\Replica') \intersect T \neq \emptyset$. Let $\Replica[q] \in (S(\Replica') \intersect T)$. As $\Replica[q]$ is good and send prepare messages for $M^{\ast}$, it must have reached the global prepare phase for $M^{\ast}$. Consequently, the condition of \lfref{fig:new_view}{recover_global_wait} holds and to complete the proof, we only need to prove that any well-formed \MName{NewViewGlobal} message will recover $M^{\ast}$.

Let $\Message{NewViewGlobal}{V, W}$ be any valid global new-view proposal for view $\View^{\ast}+1$. As $\Replica[q]$ reached the global prepare phase for $M^{\ast}$, any valid global new-view proposal must include messages from $\Shard' \in \Shards{\Transaction}$. Let $U' \subseteq \Shard'$ be the replicas in $\Shard'$ of whom messages \MName{VCGlobalStateR} are included in $W$. Let $V' = U' \difference \Faulty{\Shard'}$. We have $\abs{S'} \geq \nf{\Shard'} -\f{\Shard'}$ and $\abs{V'} \geq \nf{\Shard'} - \f{\Shard'}$. Hence, by a standard quorum argument, we conclude $S' \intersect V' \neq \emptyset'$. Let $\Replica[q]' \in (S' \intersect V')$. As $\Replica[q]'$ reached the global commit phase for $M^{\ast}$, it will meet the conditions of \lfref{fig:new_view}{global_commit} and provide both $M^{\ast}$ and a global prepare certificate for $M^{\ast}$. Let $M^{\ddag}$ be any other global preprepare certificate in $W$ accompanied by a global prepare certificate. Due to \lfref{fig:new_view}{global_high_select}, the global preprepare certificate for the newest view of $\Shard$ will be recovered. As $\View^{\ast}$ is the newest view of $\Shard$, $M^{\ddag}$ will only prevent recovery of $M^{\ast}$ if it is also a global preprepare certificate for view $\View^{\ast}$ of $\Shard$. In this case, Lemma~\ref{lem:non_divergence} guarantees that $M^{\ddag} = M^{\ast}$. Hence, any replica $\Replica$ will recover $M^{\ast}$ upon receiving $\Message{NewViewGlobal}{V, W}$.

Now assume that the induction hypothesis holds for all views $\View[j]$, $\View^{\ast} < \View[j] \leq \View[i]$. We will prove that the induction hypothesis holds for view $\View[i] + 1$. Consider any valid new-view proposal $\Message{NewView}{V}$ for view $\View[i] + 1$ and let $M^{\ddag}$ with $m(\Shard, \Transaction^{\ddag})_{\View[w], \rn} \in M^{\ddag}$ be any global preprepare certificate that is recovered due to the new-view proposal $\Message{NewView}{V}$. Hence, $M^{\ddag}$ is recovered via either \lfref{fig:new_view}{commit} or \lfref{fig:new_view}{global_commit}. In both cases, there must exist a global prepare certificate $P$ for $M^{\ddag}$. As $\Message{NewView}{V}$ is valid, we must have $\View[w] \leq \View[i]$. Hence, we can apply the second property of the induction hypothesis to conclude that $\View[w] \leq \View^{\ast}$. If $\View[w] = \View^{\ast}$, then we can use Lemma~\ref{lem:non_divergence} to conclude that $M^{\ddag} = M^{\ast}$.  Hence, to complete the proof, we must show that $\View[w] = \View^{\ast}$. First, the case in which $M^{\ddag}$ is recovered via \lfref{fig:new_view}{commit}. Due to the existence of a global commit certificate $C$ for $M^{\ddag}$, $M^{\ddag}$ satisfies the premise of this proposition. By assumption, $\View^{\ast}$ is the first view for which the premise of this proposition holds. Hence, $\View[w] \geq \View^{\ast}$, in which case we conclude $M^{\ddag} = M^{\ast}$. Last, the case in which $M^{\ddag}$ is recovered via \lfref{fig:new_view}{global_commit}. In this case, $M^{\ddag}$ is recovered via some message $\Message{NewViewGlobal}{V, W}$. Analogous to the proof for the base case, $V$ will contain a message \MName{VCRecoveryRQ} from some replica $\Replica[q] \in S(\Replica')$. Due to \lfref{fig:local_recovery}{provide}, $\Replica[q]$ will provide information on $M^{\ast}$. Consequently, a prepare certificate for $M^{\ast}$ will be obtained via global state recovery, and we also conclude $M^{\ddag} = M^{\ast}$.
\end{proof}

Lemma~\ref{lem:non_divergence} and Proposition~\ref{prop:view_change_o} are technical properties that assures that no transaction that could-be-committed by any replica will ever get lost by the system. Next, we bootstrap these technical properties to prove that all good replicas can always recover such could-be-committed transactions.

\begin{proposition}\label{prop:recovery}
Let $\Transaction$ be a transaction and $m(\Shard, \Transaction)_{\View,\rn}$ be a preprepare message. If, for all shards $\Shard^{\ast}$, $\nf{\Shard^{\ast}} > 2\f{\Shard^{\ast}}$, and there exists a shard $\Shard' \in \Shards{\Transaction}$ such that $\nf{\Shard'} - \f{\Shard'}$ good replicas in $\Shard'$ reached the global commit phase for $M$ with $m(\Shard, \Transaction)_{\View,\rn} \in M$, then every good replica in $\Shard$ will accept $M$ whenever communication becomes reliable.
\end{proposition}
\begin{proof}
Let $\Replica \in \Shard$ be a good replica that is unable to accept $M$. At some point, communication becomes reliable, after which $\Replica$ will eventually trigger \lfref{fig:recovery_main}{detect_event}. We have the following cases:
\begin{enumerate}
	\item If $\Replica$ meets the conditions of \lfref{fig:recovery_main}{force_commit}, then $\Replica$ has a local commit certificate $C(\Shard'')$, $\Shard'' \in \Shards{\Transaction}$. This local commit certificate certifies that at least $\nf{\Shard''} - \f{\Shard''}$ good replicas in $\Shard''$ finished the global prepare phase for $M$. Hence, the conditions for Proposition~\ref{prop:view_change_o} are met for $M$ and, hence, any shard in $\Shards{\Transaction}$ will maintain or recover $M$. Replica $\Replica$ can use $C(\Shard'')$ to prove this situation to other replicas, forcing them to commit to $M$, and provide any commit messages $\Replica$ is missing (\lfref{fig:recovery_main}{reply}).
	\item If $\Replica$ does not meet the conditions of \lfref{fig:recovery_main}{force_commit}, but some other good replica $\Replica[q] \in \Shard$ does, then $\Replica[q]$ can provide all missing information to $\Replica$ (\lfref{fig:local_recovery}{shortcut}). Next, $\Replica$ uses this information (\lfref{fig:local_recovery}{shortcut_reply}), after which it meets the conditions of \lfref{fig:recovery_main}{force_commit}.
    \item Otherwise, if the above two cases do not hold, then all $\nf{\Shard}$ good replicas in $\Shard$ are unable to finish the commit phase. Hence, they perform a view-change. Due to Proposition~\ref{prop:view_change_o}, this view-change will succeed and put every replica in $\Shard$ into the commit phase for $M$. As all good replicas in $\Shard$ are in the commit phase, each good replica in $\Shard$ will be able to make a local commit certificate $C(\Shard)$ for $M$, after which they meet the conditions of \lfref{fig:recovery_main}{force_commit}. \qedhere
\end{enumerate}
\end{proof}

Finally, we use Proposition~\ref{prop:recovery} to prove \emph{cross-shard-consistency}.

\begin{theorem}\label{thm:ocb_cross_cons}
Optimistic-\CB{} maintains cross-shard-consistency.
\end{theorem}
\begin{proof}
Assume a single good replica $\Replica \in \Shard$ commits or aborts a transaction $\Transaction$. Hence, it accepted some global preprepare certificate $M$ with $m(\Shard, \Transaction)_{\View,\rn} \in M$. Consequently, $\Replica$ has local commit certificates $C(\Shard')$ for $M$ of every $\Shard' \in \Shards{\Transaction}$. Hence, at least $\nf{\Shard'} -\f{\Shard'}$ good replicas in $\Shard'$ reached the global commit phase for $M$, and we can apply Proposition~\ref{prop:recovery} to conclude that any good replica $\Replica'' \in \Shard''$, $\Shard''\in \Shards{\Transaction}$ will accept $M$. As $\Replica''$ bases its commit or abort decision for $\Transaction$ on the same global prepare certificate $M$ as $\Replica$, they will both make the same decision, completing the proof.
\end{proof}

As already argued, it is straightforward to use the details of Theorem~\ref{thm:ccb} to prove that \oCB{} provides \emph{validity}, \emph{shard-involvement}, and \emph{shard-applicability}. Via Theorem~\ref{thm:ocb_cross_cons}, we proved \emph{cross-shard-consistency}. We cannot prove \emph{service} and \emph{confirmation}, however. The reason for this is simple: even though \oCB{} can detect and recover from accidental faulty behavior and accidental concurrent transactions, \oCB{} is not designed to gracefully handle targeted attacks.

\begin{example}\label{ex:concurrent_ocb}
Recall the situation of Example~\ref{ex:ccb_concurrent}. Next, we illustrate how \oCB{} deals with these concurrent transactions. We again consider distinct transactions $\Request{\Transaction_1}{\Client_1}$ and $\Request{\Transaction_2}{\Client_2}$ with $\Inputs{\Transaction_1} = \Inputs{\Transaction_2} = \{ \Object_1, \Object_2 \}$ and with $\ObjectShard{\Object_1} = \Shard_1$ and $\ObjectShard{\Object_2} = \Shard_2$. We assume that $\Shard_1$ processes $\Transaction_1$ first and $\Shard_2$ processes $\Transaction_2$ first.

The primary $\Primary{\Shard_1}$ will propose $\Transaction_1$ by prepreparing $m(\Shard_1, \Transaction_1)_{\View_1, \rn_1}$. In doing so, $\Primary{\Shard_1}$ sends $m(\Shard_1, \Transaction_1)_{\View_1, \rn_1}$ to all replicas $\Shard_1 \union \Shard_2$. Next, the replicas in $\Shard_1$ will wait for a message $m(\Shard_2, \Transaction_1)_{\View_2', \rn_2'}$ from $\Shard_2$. At the same time,  $\Primary{\Shard_2}$ already proposed $\Transaction_2$ at the same time by sending out $m(\Shard_2, \Transaction_2)_{\View_2, \rn_2}$, and the replicas in $\Shard_2$ will wait for a message $m(\Shard_1, \Transaction_2)_{\View_1', \rn_1'}$. Hence, the replicas in $\Shard_1$ will never receive $m(\Shard_2, \Transaction_1)_{\View_2', \rn_2'}$ and the replicas in $\Shard_2$ will never receive $m(\Shard_1, \Transaction_2)_{\View_1', \rn_1'}$. Consequently, no replica will finish global preprepare, the consensus round will fail for all replicas, and all good replicas will initiate a view-change. As no replica reached the global prepare phase, transactions $\Transaction_1$ and $\Transaction_2$ do not need to be recovered during the view-change. After the view-changes, both $\Shard_1$ and $\Shard_2$ can process other transactions (or retry $\Transaction_1$ or $\Transaction_2$), but if they both process $\Transaction_1$ and $\Transaction_2$ again, the system will again initiate a view-change.
\end{example}

As said before, \oCB{} is optimistic in the sense that it is optimized for the situation in which faulty behavior (including concurrent transactions) is rare. Still, in all cases, \oCB{} maintains cross-shard consistency, however. Moreover, in the optimistic case in which shards have good primaries and no concurrent transactions exist, progress is guaranteed whenever communication is reliable:

\begin{proposition}
If, for all shards $\Shard^{\ast}$, $\nf{\Shard^{\ast}} > 2\f{\Shard^{\ast}}$, and Assumptions~\ref{ass:coordinate}, \ref{ass:inout}, \ref{ass:wellclient}, and \ref{ass:shard_minimality} hold, then Optimistic-\CB{} satisfies Requirements~\ref{req:validity}--\ref{req:shard_confirm} in the optimistic case.
\end{proposition}

If the optimistic assumption does not hold, then this can result in coordinated attempts to prevent \oCB{} from making progress. At the core of such attacks is the ability for malicious clients and malicious primaries to corrupt the operations of shards coordinated by good primaries, as already shown in Example~\ref{ex:ocb_failures}.  To reduce the impact of targeted attacks, one can opt to make primary election non-deterministic, e.g., by using shard-specific distributed coins to elect new primaries in individual shards~\cite{coin,bdrandcoin}.

As a final note, we remark that we have presented \oCB{} with a per-round checkpoint and recovery method. In this simplified design, the recovery path only has to recover at-most a single round. Our approach can easily be generalized to a more typical multi-round checkpoint and recovery method, however. Furthermore, we believe that the way in which \oCB{} extends \Pbft{} can easily be generalized to other consensus protocols, e.g., \HS{}.

\section{Pessimistic-\CB{}: transaction processing under attack}\label{sec:pcb}

In the previous section, we introduced \oCB{}, a general-purpose minimalistic and efficient multi-shard transaction processing protocol. \oCB{} is designed with the assumption that malicious behavior is rare, due to which it can minimize coordination in the normal-case while requiring intricate coordination when recovering from attacks. As an alternative to the optimistic approach of \oCB{},  we can apply a \emph{pessimistic} approach to \cCB{} to gracefully recover from concurrent transactions that is geared towards minimizing the influence of malicious behavior altogether. Next, we explore such a pessimistic design via \emph{Pessimistic-\CB} (\pCB{}). 

The design of \pCB{} builds upon the design of \cCB{} by adding additional coordination to the cross-shard exchange and decide outcome steps. As in \cCB{}, the acceptance of $m(\Shard, \Transaction)_{\rn}$ in round $\rn$ by all good replicas completes the \emph{local inputs} step. Before cross-shard exchange, the replicas in $\Shard$ destruct the objects in $D(\Shard, \Transaction)$, thereby fully pledging these objects to $\Transaction$ until the commit or abort decision. Then, $\Shard$ performs cross-shard exchange by broadcasting $m(\Shard, \Transaction)_{\rn}$ to all other shards in $\Shards{\Transaction}$, while the replicas in $\Shard$ wait until they receive messages $m(\Shard', \Transaction)_{\rn'} = (\Request{\Transaction}{\Client}, I(\Shard', \Transaction), D(\Shard', \Transaction))$ from all other shards $\Shard' \in \Shards{\Transaction}$.

After cross-shard exchange comes the final \emph{decide outcome} step. After $\Shard$ receives $m(\Shard', \Transaction)_{\rn'}$ from all shards $\Shard' \in \Shards{\Transaction}$, the replicas force a \emph{second consensus step} that determines the round $\rn^{\ast}$ at which $\Shard$ decides \emph{commit} (whenever  $I(\Shard', \Transaction) = D(\Shard', \Transaction)$ for all $\Shard' \in \Shards{\Transaction}$) or \emph{abort}. If $\Shard$ decides commit, then, in round $\rn^{\ast}$, all good replicas in $\Shard$ construct all objects $\Object \in \Outputs{\Transaction}$ with $\Shard = \ObjectShard{\Object}$. If $\Shard$ decides abort, then, in round $\rn^{\ast}$, all good replicas in $\Shard$ reconstruct all objects in $D(\Shard, \Transaction)$ (rollback). Finally, each good replica informs $\Client$ of the outcome of execution. If $\Client$ receives, from every shard $\Shard' \in \Shards{\Transaction}$, identical outcomes from $\nf{\Shard'} - \f{\Shard}$ distinct replicas in $\Shard'$, then it considers $\Transaction$ to be successfully executed. In Figure~\ref{fig:flow_pcb}, we sketched the working of \pCB{}.

\begin{figure}[t!]
    \centering
    \scalebox{0.825}{
    \begin{tikzpicture}[yscale=0.625,xscale=1.45,every edge/.append style={thick},label/.append style={below=5pt,align=center,font=\footnotesize}]
        \node[left] (c) at  (0.75, 3.5) {$\Client$};
        \node[left] (s1) at (0.75, 2) {$\Shard_1$};
        \node[left] (s2) at (0.75, 1) {$\Shard_2$};
        \node[left] (s3) at (0.75, 0) {$\Shard_3$};
        
        \path (0.75, 0) edge ++(9.45, 0)
              (0.75, 1) edge ++(9.45, 0)
              (0.75, 2) edge ++(9.45, 0)
              (0.75, 3.5) edge[green!70!black] ++(9.45, 0);

        \path (1, 0) edge ++(0, 3.5)
              (2, 0) edge ++(0, 3.5)
              (5, 0) edge ++(0, 3.5)
              (6, 0) edge ++(0, 3.5)
              (9, 0) edge ++(0, 3.5)
              (10, 0) edge ++(0, 3.5);

        \path[->] (1, 3.5) edge node[above] {$\Request{\Transaction}{\Client}$} (2, 2) edge  (2, 1) edge  (2, 0)
                  (5, 2) edge (6, 1) edge (6, 0)
                  (5, 1) edge (6, 2) edge (6, 0)
                  (5, 0) edge (6, 2) edge (6, 1)
                  (9, 0) edge (10, 3.5)
                  (9, 1) edge (10, 3.5)
                  (9, 2) edge (10, 3.5);
        
        \draw[draw=orange!60,fill=orange!40,rounded corners] (2.1, -0.4) rectangle (4.9, 0.4);
        \draw[draw=orange!60,fill=orange!40,rounded corners] (2.1,  0.6) rectangle (4.9, 1.4);
        \draw[draw=orange!60,fill=orange!40,rounded corners] (2.1,  1.6) rectangle (4.9, 2.4);
        \node at (3.5, 0) {Consensus on $\Request{\Transaction}{\Client}$};
        \node at (3.5, 1) {Consensus on $\Request{\Transaction}{\Client}$};
        \node at (3.5, 2) {Consensus on $\Request{\Transaction}{\Client}$};

        \draw[draw=orange!60,fill=orange!40,rounded corners] (6.1, -0.4) rectangle (8.9, 0.4);
        \draw[draw=orange!60,fill=orange!40,rounded corners] (6.1,  0.6) rectangle (8.9, 1.4);
        \draw[draw=orange!60,fill=orange!40,rounded corners] (6.1,  1.6) rectangle (8.9, 2.4);
        \node at (7.5, 0) {Consensus on Commit/Abort};
        \node at (7.5, 1) {Consensus on Commit/Abort};
        \node at (7.5, 2) {Consensus on Commit/Abort};
        
        \node[label] at (3.5, 0) {Local Inputs\\(Consensus)};
        \node[label] at (5.5, 0) {Cross-Shard Exchange\\(Cluster-Sending)};
        \node[label] at (7.5, 0) {Decide Outcome\\(Consensus)};
        \node[label] at (9.5, 0) {Inform};
        
        \node[label,above=5pt] at (5, 3.5) {\emph{destruction}};
        \node[label,above=5pt] at (9, 3.5) {\emph{construction} or \emph{rollback}};
    \end{tikzpicture}
    }
    \caption{The message flow of \pCB{} for a $3$-shard client request $\Request{\Transaction}{\Client}$ that is committed.}\label{fig:flow_pcb}
\end{figure}

We notice that processing a multi-shard transaction via \pCB{} requires \emph{two} consensus steps per shard. In some cases, we can eliminate the second step, however. First, if $\Transaction$ is a multi-shard transaction with $\Shard \in \Shards{\Transaction}$ and the replicas in $\Shard$ accept $(\Request{\Transaction}{\Client}, I(\Shard, \Transaction), D(\Shard, \Transaction))$ with $I(\Shard, \Transaction) \neq D(\Shard, \Transaction)$, then the replicas can immediately abort whenever they accept $(\Request{\Transaction}{\Client}, I(\Shard, \Transaction), D(\Shard, \Transaction))$. Second, if $\Transaction$ is a single-shard transaction with $\Shards{\Transaction} = \{\Shard\}$, then the replicas in $\Shard$ can immediately decide commit or abort whenever they accept $(\Request{\Transaction}{\Client}, I(\Shard, \Transaction), D(\Shard, \Transaction))$. Hence, in both cases, processing of $\Transaction$ at $\Shard$ only requires a single consensus step at $\Shard$.

Next, we illustrate how \pCB{} deals with concurrent transactions.

\begin{example}\label{ex:concurrent_pcb}
Recall the situation of Example~\ref{ex:ccb_concurrent}. Next, we illustrate how \pCB{} deals with these concurrent transactions. We again consider distinct transactions $\Request{\Transaction_1}{\Client_1}$ and $\Request{\Transaction_2}{\Client_2}$ with $\Inputs{\Transaction_1} = \Inputs{\Transaction_2} = \{ \Object_1, \Object_2 \}$ and with $\ObjectShard{\Object_1} = \Shard_1$ and $\ObjectShard{\Object_2} = \Shard_2$. We assume that $\Shard_1$ processes $\Transaction_1$ first and $\Shard_2$ processes $\Transaction_2$ first.

Shard $\Shard_1$ will start by \emph{destructing} $\Object_1$ and sends $(\Request{\Transaction_1}{\Client_1}, \{ \Object_1 \}, \{ \Object_1 \})$ to $\Shard_2$. Next, $\Shard_1$ will wait, during which it receives $\Transaction_2$. At the same time, $\Shard_2$ follows similar steps for $\Transaction_2$ and sends $(\Request{\Transaction_2}{\Client_2}, \{ \Object_2 \}, \{ \Object_2 \})$ to $\Shard_1$. While $\Shard_1$ is waiting for information on $\Transaction_1$ from $\Shard_2$, it receives $\Transaction_2$ and starts processing of $\Transaction_2$. Shard $\Shard_1$ directly determines that $\Object_1$ does no longer  exist. Hence, it sends $(\Request{\Transaction_2}{\Client_2}, \{ \Object_1 \}, \emptyset \})$ to $\Shard_2$. Likewise, $\Shard_2$ will start processing of $\Transaction_1$, sending $(\Request{\Transaction_1}{\Client_1}, \{ \Object_2 \}, \emptyset)$ to $\Shard_1$ as a result.

After the above exchange, both $\Shard_1$ and $\Shard_2$ conclude that transactions $\Transaction_1$ and $\Transaction_2$ must be aborted, which they eventually both do, after which $\Object_1$ is restored in $\Shard_1$ and $\Object_2$ is restored in $\Shard_2$.
\end{example}

We notice that this situation leads to \emph{both} transactions being aborted. Furthermore, we see that even though transactions get aborted, individual replicas can all determine whether their shard performed the necessary steps and, hence, whether their primary operated correctly. Next, we prove the correctness of \pCB{}:

\begin{theorem}\label{thm:pcb}
If, for all shards $\Shard^{\ast}$, $\nf{\Shard^{\ast}} > 2\f{\Shard^{\ast}}$, and Assumptions~\ref{ass:coordinate}, \ref{ass:inout}, \ref{ass:wellclient}, and \ref{ass:shard_minimality} hold, then Pessimistic-\CB{} satisfies Requirements~\ref{req:validity}--\ref{req:shard_confirm}.
\end{theorem}
\begin{proof}
Let $\Transaction$ be a transaction. As good replicas in $\Shard$ discard $\Transaction$ if it is invalid or if $\Shard \notin \Shards{\Transaction}$,  \pCB{} provides \emph{validity} and \emph{shard-involvement}. Next, \emph{shard-applicability} follow directly from the decide outcome step.

If a shard $\Shard$ commits or aborts transaction $\Transaction$, then it must have completed the decide outcome and cross-shard exchange steps. As $\Shard$ completed cross-shard exchange, all shards $\Shard' \in \Shards{\Transaction}$ must have exchanged the  necessary information to $\Shard$. By relying on cluster-sending for cross-shard exchange, $\Shard'$ requires cooperation of all good replicas in $\Shard'$ to exchange the necessary information to $\Shard$. Hence, we have the guarantee that these good replicas will also perform cross-shard exchange to any other shard $\Shard'' \in \Shards{\Transaction}$. Hence, every shard $\Shard'' \in \Shards{\Transaction}$ will receive the same information as $\Shard$, complete cross-shard exchange, and make the same decision during the decide outcome step, providing \emph{cross-shard consistency}.

A client can force service on a transaction $\Transaction$ by choosing a shard $\Shard \in \Shards{\Transaction}$ and sending $\Transaction$ to all good replicas in $\NonFaulty{\Shard}$. By doing so, the normal mechanisms of consensus can be used by the good replicas in $\NonFaulty{\Shard}$ to force acceptance on $\Transaction$ in $\Shard$ and, hence, bootstrapping acceptance on $\Transaction$ in all shards $\Shard' \in \NonFaulty{\Shard}$. Due to cross-shard consistency,  every shard in $\Shards{\Transaction}$ will perform the necessary steps to eventually inform the client. As all good replicas $\Replica \in \Shard$, $\Shard \in \Shards{\Transaction}$, will inform the client of the outcome for $\Transaction$, the majority of these inform-messages come from good replicas, enabling the client to reliably derive the true outcome. Hence, \cCB{} provides \emph{service} and \emph{confirmation}.
\end{proof}

\section{The strengths of \CB{}}\label{sec:compare}

Having introduced the three variants of \CB{} in Sections~\ref{sec:ccb},~\ref{sec:ocb}, and~\ref{sec:pcb}, we will now analyze the strengths and performance characteristics of each of the variants. First, we will show that \CB{} provides serializable execution~\cite{isolation,ansisql}. Second, we look at the ability of \CB{} to maximize per-shard throughput by supporting out-of-order processing. Finally, we compare the costs, the attainable performance, and the scalability of the three protocols.

\subsection{The ordering of transactions in \CB{}}
The data model utilized by \cCB{}, \oCB{}, and \pCB{} guarantees that any object $\Object$ can only be involved in at-most \emph{two} committed transactions: one that \emph{constructs} $\Object$ and another one that \emph{destructs} $\Object$. Assume the existence of such transactions $\Transaction_1$ and $\Transaction_2$ with $\Object \in \Outputs{\Transaction_1}$ and $\Object \in \Inputs{\Transaction_2}$. Due to \emph{cross-shard-consistency} (Requirement~\ref{req:shard_const}), the shard $\ObjectShard{\Object}$ will have to execute both $\Transaction_1$ and $\Transaction_2$. Moreover, due to \emph{shard-applicability} (Requirement~\ref{req:shard_app}), the shard $\ObjectShard{\Object}$ will execute $\Transaction_1$ strictly before $\Transaction_2$. Now consider the relation 
\[
 \TOrder \mathrel{:=} \{ (\Transaction, \Transaction') \mid 
(\text{the system committed to $\Transaction$ and $\Transaction'$}) \land
    (\Outputs{\Transaction} \intersect \Inputs{\Transaction'} \neq \emptyset) \}.
\]
Obviously, we have $\TOrder(\Transaction_1, \Transaction_2)$. Next, we will prove that all committed transactions are executed in a \emph{serializable} ordering. As a first step, we prove the following:

\begin{lemma}
If we interpret transactions as nodes and $\TOrder$ as an edge relation, then the resulting graph is \emph{acyclic}.
\end{lemma}
\begin{proof}
The proof is by contradiction. Let $G$ be the graph-interpretation of $\TOrder$. We assume that graph $G$ is cyclic. Hence, there exists transactions $\Transaction_0, \dots, \Transaction_{m-1}$ such that $\TOrder(\Transaction_i, \Transaction_{i+1})$, $0 \leq i < m-1$, and $\TOrder(\Transaction_{m-1}, \Transaction_0)$. By the definition of $\TOrder$, we can choose objects $\Object_i$, $0 \leq i < m$, with $\Object_i \in (\Outputs{\Transaction_i} \intersect \Inputs{\Transaction_{(i+1) \bmod m}})$. Due to \emph{cross-shard-consistency} (Requirement~\ref{req:shard_const}), the shard $\ObjectShard{\Object_i}$, $0 \leq i < m$, executed transactions $\Transaction_i$ and $\Transaction_{(i+1) \bmod m}$.

Consider $\Object_i$, $0 \leq i < m$, and let $t_i$ be the time at which shard $\ObjectShard{\Object_i}$ executed $\Transaction_i$ and constructed $\Object_i$. Due to \emph{shard-applicability} (Requirement~\ref{req:shard_app}), we know that shard $\ObjectShard{\Object_i}$ executed $\Transaction_{(i+1) \bmod m}$ strictly after $t_i$. Moreover, also shard $\ObjectShard{\Object_{(i+1) \bmod m}}$  must have executed $\Transaction_{(i+1) \bmod m}$ strictly after $t_i$ and we derive $t_i < t_{(i+1)\bmod m}$. Hence, we must have $t_0 < t_1 < \dots < t_{m-1} < t_0$, a contradiction. Consequently, $G$ must be acyclic.
\end{proof}

To derive a serializable execution order for all committed transactions, we simply construct a directed acyclic graph in which transactions are nodes and $\TOrder$ is the edge relation. Next, we \emph{topologically sort} the graph to derive the searched-for ordering. Hence, we conclude:

\begin{theorem}
A sharded fault-tolerant system that uses the object-dataset data model, processes UTXO-like transactions, and satisfies Requirements~\ref{req:validity}-\ref{req:shard_service} commits transactions in a serializable order.
\end{theorem}

We notice that \CB{} only provides serializability for \emph{committed} transactions. As we have seen in Example~\ref{ex:concurrent_pcb}, concurrent transactions are not executed in a serializable order, as they are aborted. It is this flexibility in dealing with aborted transactions that allows all variants of \CB{} to operate with minimal and fully-decentralized coordination between shards; while still providing strong isolation for all committed transactions. 

\subsection{Out-of-order processing in \CB{}}

In normal consensus-based systems, the \emph{latency} for a single consensus decision is ultimately determined by the message delay $\delta$. E.g., with the three-phase design of \Pbft{}, it will take at least $3\delta$ before a transaction that arrives at the primary is executed by all replicas. To minimize the influence of message delay on \emph{throughput}, some consensus-based systems support \emph{out-of-order decision making} in which the primary is allowed to maximize bandwidth usage by continuously proposing transactions for future rounds (while previous rounds are processed by the replicas). To illustrate this, one can look at fine-tuned implementations of \Pbft{} running at replicas that have sufficient memory buffers available~\cite{pbftj,icdcs}. In this setting, replicas can work on several consensus rounds at the same time by allowing the primary to propose for rounds within a window of rounds.

As the goal of \CB{} is to maximize performance---both in terms of latency (\oCB{}) and in terms of throughput---we have designed \CB{} to support out-of-order processing (if provided by the underlying consensus protocol, in the case of \cCB{} and \pCB{}). The only limitation to these out-of-order processing capabilities are with respect to transactions affecting a shared object: such transactions must be proposed strictly in-order, as otherwise the set of pledged inputs cannot be correctly determined by the good replicas. This is not a limitation for the normal-case operations, however, as such concurrent transactions only happen due to malicious behavior.

\subsection{A comparison of the three \CB{} variants}

Finally, we compare the practical costs of the three \CB{} multi-shard transaction processing protocols. First, in Figure~\ref{fig:compare_cb}, we provide a high-level comparison of the costs of each of the protocols to process a single transaction $\Transaction$ that affects $s = \abs{\Shards{\Transaction}}$ distinct shards. For the normal-case behavior, we compare the complexity in the number of \emph{consensus steps} per shard and the number of \emph{cross-shard exchange} steps between shards (which together determine the maximum throughput), and the number of \emph{consecutive communication phases} (which determines the minimum latency).

\begin{figure}[t!]
    \centering
    \small
    \setlength\tabcolsep{2pt}
    \begin{tabular}{l|ccc|cc}
                &\multicolumn{3}{c|}{Normal-case complexity}&Concurrent&View-changes\\
    Protocol name&Consensus&Exchange&Phases&Transactions&\\
    \hline
    \hline
    \cCB{} & $s$  & $1$ & $4$ & Objects pledged & Single-shard\\
    \oCB{} & $s$  & $3$ & $3$ & View-change \& Abort & Multi-shard\\
    \pCB{} & $2s$ & $1$ & $7$ & Normal-case Abort& Single-shard\\
    \end{tabular}
    \caption{Comparison of the three \CB{} protocols for processing a transaction that affects $s$ shards. We compare the normal-case complexity, how they deal with concurrent transactions (due to malicious clients), and how they deal with malicious primaries.}\label{fig:compare_cb}
\end{figure}

Next, we compare how the three protocols deal with malicious behavior by clients and by replicas. If no clients behave malicious, then all transactions will \emph{commit}. In all three protocols, malicious behavior by clients can lead to the existence of concurrent transactions that affect the same object. Upon detection of such concurrent transactions, all three protocols will \emph{abort}. The consequences of such an abort are different in the three protocols.

In \cCB{}, objects affected by aborted transactions remain pledged and cannot be reused. In practice, this loss of objects can provide an incentive for clients to not behave malicious, but does limit the usability of \cCB{} in non-incentivized environments. \oCB{} is optimized with the assumption that conflicting concurrent transactions are rare. When conflicts occur, they can lead to the failure of a global consensus round, which can lead to a view-change in one or more affected shards (even if none of the primaries is faulty). Finally, \pCB{} deals with concurrent transactions by aborting them via the normal-case of the protocol. To be able to do so, \pCB{} does require additional consensus steps, however.

The three \CB{} protocols are resilient against malicious replicas: only malicious primaries can affect the normal-case operations of these protocols. If malicious primaries behave sufficiently malicious to affect the normal-case operations, their behavior is detected, and the primary is replaced. In both \cCB{} and \pCB{}, dealing with a malicious primary in a shard can be done completely in isolation of all other shards. In \oCB{}, which is optimized with the assumption that failures are rare, the failure of a primary while processing a transaction $\Transaction$ can lead to view-changes in all shards affected by $\Transaction$.

\pgfplotstableread{
num_values_per_txn	num_shards	num_shard_steps	num_shard_steps_avg_shard	pbft_tput	ccb_tput	ocb_tput	pcb_tput
2	1	16777216	16777216	1252	1252	1252	1252
2	2	25160852	12580426	1252	972	943	700
2	4	29356693	7339173	1252	1488	1439	1029
2	8	31454753	3931844	1252	2664	2572	1818
2	16	32503483	2031467	1252	5063	4885	3435
2	32	33028980	1032155	1252	9880	9529	6685
2	64	33292028	520187	1252	19521	18827	13193
2	128	33423112	261118	1252	38810	37427	26213
2	256	33488808	130815	1252	77390	74628	52254
2	512	33522006	65472	1252	154547	149030	104335
2	1024	33538283	32752	1252	308867	297839	208501
2	2048	33546458	16380	1252	617506	595455	416832
2	4096	33550528	8191	1252	1234786	1190688	833495
2	8192	33552503	4095	1252	2469351	2381162	1666827
2	16384	33553497	2047	1252	4938480	4762108	3333488
}\dataTwo

\pgfplotstableread{
num_values_per_txn	num_shards	num_shard_steps	num_shard_steps_avg_shard	pbft_tput	ccb_tput	ocb_tput	pcb_tput
4	1	16777216	16777216	921	921	921	921
4	2	31457684	15728842	921	440	473	310
4	4	45874978	11468744	921	386	425	300
4	8	55543071	6942883	921	527	585	426
4	16	61076149	3817259	921	882	980	724
4	32	64029162	2000911	921	1617	1800	1337
4	64	65553594	1024274	921	3099	3451	2571
4	128	66327362	518182	921	6067	6760	5042
4	256	66717735	260616	921	12007	13379	9985
4	512	66913035	130689	921	23889	26620	19874
4	1024	67010978	65440	921	47653	53103	39653
4	2048	67060121	32744	921	95181	106068	79209
4	4096	67084884	16378	921	190236	211998	158322
4	8192	67096820	8190	921	380350	423863	316550
4	16384	67102828	4095	921	760579	847592	633006
}\dataFour

\pgfplotstableread{
num_values_per_txn	num_shards	num_shard_steps	num_shard_steps_avg_shard	pbft_tput	ccb_tput	ocb_tput	pcb_tput
8	1	16777216	16777216	602	602	602	602
8	2	33423730	16711865	602	237	282	170
8	4	60392383	15098095	602	129	164	108
8	8	88103683	11012960	602	116	151	103
8	16	108261010	6766313	602	154	201	139
8	32	120423859	3763245	602	250	328	228
8	64	127105399	1986021	602	450	591	413
8	128	130605790	1020357	602	854	1123	787
8	256	132397993	517179	602	1665	2190	1534
8	512	133304475	260360	602	3287	4324	3031
8	1024	133760291	130625	602	6532	8593	6024
8	2048	133988747	65424	602	13022	17131	12010
8	4096	134103708	32740	602	26002	34207	23983
8	8192	134161118	16377	602	51963	68360	47930
8	16384	134189728	8190	602	103883	136665	95822
}\dataEight

\pgfplotstableread{
num_values_per_txn	num_shards	num_shard_steps	num_shard_steps_avg_shard	pbft_tput	ccb_tput	ocb_tput	pcb_tput
16	1	16777216	16777216	356	356	356	356
16	2	33553959	16776979	356	130	166	95
16	4	66436677	16609169	356	58	81	50
16	8	118372927	14796615	356	34	50	31
16	16	172851484	10803217	356	31	46	30
16	32	213822319	6681947	356	41	61	39
16	64	239155878	3736810	356	66	98	63
16	128	253264062	1978625	356	118	175	113
16	256	260709675	1018397	356	223	331	215
16	512	264540336	516680	356	433	645	418
16	1024	266480101	260234	356	855	1271	824
16	2048	267455786	130593	356	1698	2525	1637
16	4096	267946341	65416	356	3384	5033	3263
16	8192	268191059	32738	356	6756	10049	6515
16	16384	268313976	16376	356	13499	20081	13018
}\dataSixteen

\pgfplotstableread{
num_values_per_txn	num_shards	num_shard_steps	num_shard_steps_avg_shard	pbft_tput	ccb_tput	ocb_tput	pcb_tput
32	1	16777216	16777216	196	196	196	196
32	2	33554432	16777216	196	68	91	50
32	4	67102106	16775526	196	29	44	25
32	8	132349287	16543660	196	14	22	13
32	16	234403684	14650230	196	8	14	8
32	32	342490153	10702817	196	8	13	8
32	64	425045691	6641338	196	10	17	10
32	128	476662676	3723927	196	16	27	16
32	256	505610360	1975040	196	30	48	29
32	512	520937896	1017456	196	56	91	55
32	1024	528829360	516434	196	110	177	108
32	2048	532833188	260172	196	217	349	214
32	4096	534850826	130578	196	432	693	424
32	8192	535864872	65413	196	861	1381	846
32	16384	536372624	32737	196	1718	2757	1689
}\dataThirtyTwo

\pgfplotstableread{
num_values_per_txn	num_shards	num_shard_steps	num_shard_steps_avg_shard	pbft_tput	ccb_tput	ocb_tput	pcb_tput
64	1	16777216	16777216	103	103	103	103
64	2	33554432	16777216	103	35	48	26
64	4	67108863	16777215	103	15	23	13
64	8	134191639	16773954	103	7	11	6
64	16	264121962	16507622	103	3	5	3
64	32	466507850	14578370	103	2	3	2
64	64	681851502	10653929	103	2	3	2
64	128	847544877	6621444	103	2	4	2
64	256	951701595	3717584	103	4	7	4
64	512	1010291333	1973225	103	7	12	7
64	1024	1041393177	1016985	103	14	24	14
64	2048	1057415383	516316	103	27	46	27
64	4096	1065553859	260144	103	54	91	54
64	8192	1069653436	130572	103	108	182	108
64	16384	1071710365	65412	103	217	363	215
}\dataSixtyFour

\begin{figure}[t!]
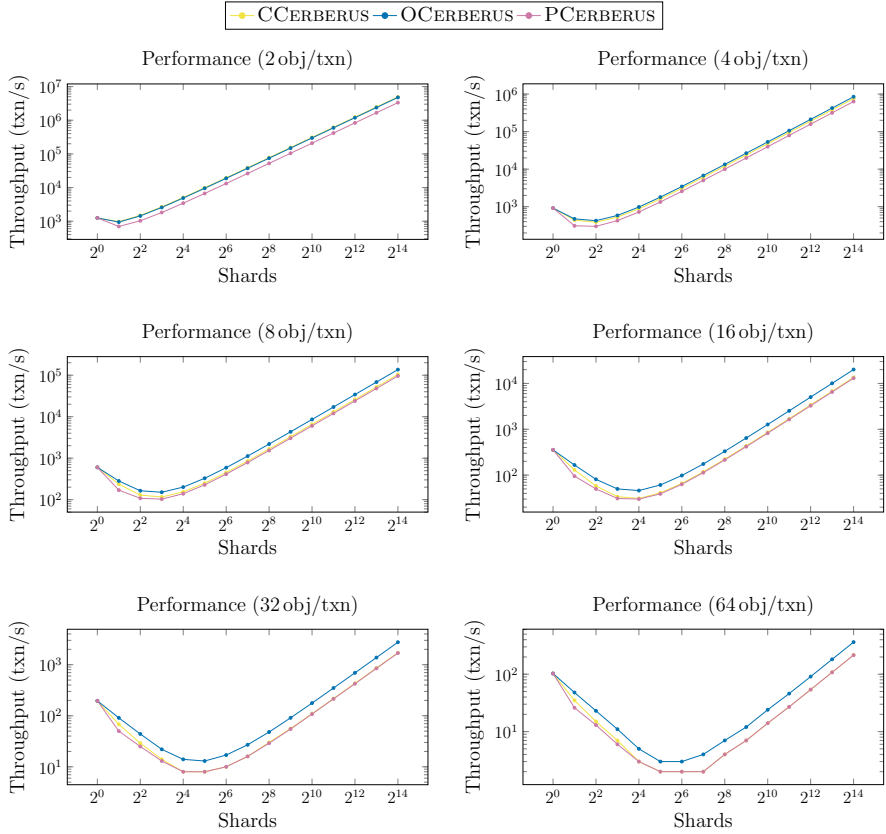

    \centering
    \scalebox{0.535}{\ref*{mainlegend}}\\[5pt]
    \begin{tabular}{ccc}
        \plotPerformance{\dataTwo}{2}{legend to name={mainlegend},legend columns=-1,legend entries={\cCB{},\oCB{},\pCB{}}}&
        \plotPerformance{\dataFour}{4}{}\\
        \\
        \plotPerformance{\dataEight}{8}{}&
        \plotPerformance{\dataSixteen}{16}{}\\
        \\
        \plotPerformance{\dataThirtyTwo}{32}{scaled y ticks=base 10:-3}&
        \plotPerformance{\dataSixtyFour}{64}{scaled y ticks=base 10:-3}
    \end{tabular}
    \caption{Throughput of the three \CB{} protocols as a function of the number of shards.}\label{fig:model_tput}
\end{figure}

Finally, we illustrate the performance of \CB{}. To do so, we have modeled the maximum throughput of each of these protocols in an environment where each shard has seven replicas (of which two can be faulty) and each replica has a bandwidth of $\SI{100}{\mega\bit\per\second}$. We have chosen to optimize \cCB{}, \oCB{}, and \pCB{} to minimize \emph{processing latencies} over minimizing bandwidth usage (e.g., we do not batch requests and the cross-shard exchange steps do not utilize threshold signatures; with these techniques in place we can boost throughput by a constant factor at the cost of the per-transaction processing latency). In Figure~\ref{fig:model_tput}, we have visualized the maximum attainable throughput for each of the protocols as function of the number of shards. In Figure~\ref{fig:model_steps}, we have visualized the number of per-shard steps performed by the system (for \cCB{} and \oCB{}, this is equivalent to the number of per-shard consensus steps, for \pCB{} this is half the number of per-shard consensus steps). As one can see from these figures, all three protocols have excellent scalability: increasing the number of shards will increase the overall throughput of the system. Sharding does come with clear overheads, however, increasing the number of shards also increases the number of shards affected by each transaction, thereby increasing the overall number of consensus steps. This is especially true for very large transactions that affect many objects (that can affect many distinct shards).

\begin{figure}[t!]
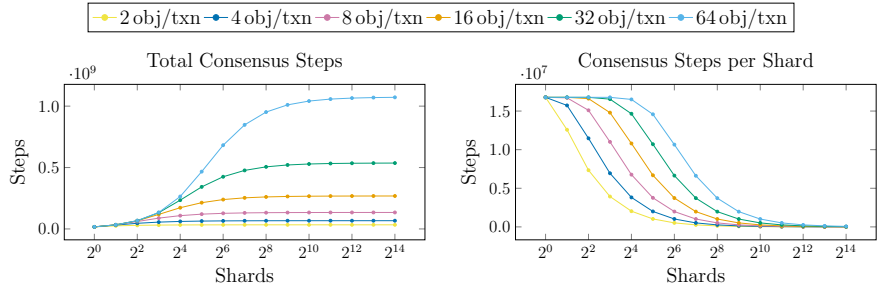

    \centering
    \scalebox{0.535}{\ref*{sizelegend}}\\[5pt]
    \begin{tabular}{cc}
        \plotSizes{num_shard_steps}{Total Consensus Steps}{
                legend to name={sizelegend},
                legend columns=-1,
                legend entries={\SI{2}{\text{obj}\per\text{txn}},
                                \SI{4}{\text{obj}\per\text{txn}},
                                \SI{8}{\text{obj}\per\text{txn}},
                                \SI{16}{\text{obj}\per\text{txn}},
                                \SI{32}{\text{obj}\per\text{txn}},
                                \SI{64}{\text{obj}\per\text{txn}}}}&
        \plotSizes{num_shard_steps_avg_shard}{Consensus Steps per Shard}{}
    \end{tabular}
    \caption{Amount of work, in terms of consensus steps, for the shards involved in processing the transactions.}\label{fig:model_steps}
\end{figure}

\section{Related Work}\label{sec:related}

Distributed systems are typically employed to either increase reliability (e.g., via consensus-based fault-tolerance) or to increase performance (e.g., via sharding). Consequently, there is abundant literature on such distributed systems, distributed databases, and sharding (e.g.,~\cite{distdb,distbook,distalgo}) and on consensus-based fault-tolerant systems (e.g.,~\cite{distalgo,wild,scaling,untangle,encybd}). Next, we shall focus on the few works that deal with sharding in fault-tolerant systems.

Several recent system papers have proposed specialized systems that combine sharding with consensus-based resilient systems. Examples include \Name{AHL}~\cite{ahl}, \Name{Caper}~\cite{caper}, \Name{Chainspace}~\cite{chainspace}, and \Name{SharPer}~\cite{sharper}, which all use sharding for data management and transaction processing. Systems such as \Name{AHL} and \Name{Caper} are designed with single-shard workloads in mind, as they rely on centralized orderers to order and process multi-shard transactions, whereas systems such as \Name{Chainspace} and \Name{SharPer} are closer to the decentralized design of \CB{}. In specific, \Name{Chainspace} uses a consensus-based commit protocol that performs three consecutive consensus and cross-shard exchange steps that resemble the two-step approach of \pCB{} (although the details of the recovery path are rather different). In comparison, \CB{} greatly improves on the design of \Name{Chainspace} by reducing the number of consecutive consensus steps necessary to process transactions and by introducing out-of-order transaction processing capabilities. Finally, \Name{SharPer} integrates global consensus steps in a consensus protocol in a similar manner as \oCB{}. Their focus is mainly on a crash-tolerant \Name{Paxos} protocol, however, and they do not fully explorer the details of a full Byzantine fault-tolerant recovery path.

A few fully-replicate consensus-based systems utilize sharding at the level of consensus decision making, this to improve consensus throughput~\cite{geobft,steward,algorand,omada}. In these systems, only a small subset of all replicas, those in a single shard, participate in the consensus on any given transaction, thereby reducing the costs to replicate this transaction without improving storage and processing scalability.  Finally, the recently-proposed \emph{delayed-replication algorithm} aims at improving scalability of resilient systems by separating fault-tolerant data storage from specialized data processing tasks~\cite{delay_rep}, the latter of which can be distributed over many participants.

Recently, there has also been promising work on sharding and techniques supporting sharding for permissionless blockchains. Examples include techniques to enable sidechains, blockchain relays, and atomic swaps~\cite{blockchaindb,btcrelay,atomswap,cosmos,polkadot}, which each enable various forms of cooperation between blockchains (including simple cross-chain communication and cross-chain transaction coordination). Unfortunately, these permissionless techniques are several orders of magnitudes slower than comparable techniques for traditional fault-tolerant systems, making them incomparable with the design of \CB{} discussed in this work.

\section{Conclusion}\label{sec:concl}   
In this paper, we introduced Core-\CB{}, Optimistic-\CB{}, and Pessimistic-\CB{}, three fully distributed approaches towards multi-shard fault-tolerant transaction processing. The design of these approaches is geared towards processing UTXO-like transactions in sharded distributed ledger networks with minimal cost, while maximizing performance. By using the properties of UTXO-like transactions to our advantage, both Core-\CB{} and Optimistic-\CB{} are optimized for cases with fewer expected malicious behaviors, in which case they are able to provide serializable transaction processing with only a single consensus step per affected shard, whereas Pessimistic-\CB{} is optimized to efficiently deal with a broad-range of malicious behavior at the cost of a second consensus step during normal operations.

The core ideas of \CB{} are not tied to any particular underlying consensus protocol. In this work, we have chosen to build \CB{} on top of \Pbft{}, as our experience shows that well-tuned implementations that use \emph{out-of-order processing} of this protocol can outperform most other protocols in raw throughput~\cite{icdcs}. Combining other consensus protocols with \CB{} will result in other trade-offs between maximum throughput, per-transaction latency, bandwidth usage, and (for protocols that do not support out-of-order processing) vulnerability to message delays. Applying the ideas of \CB{} fully onto other consensus protocols in a fully fine-tuned manner remains open, however. E.g., we are very interested in seeing whether incorporating \CB{} into the more-resilient four-phase design of \HS{} can sharply reduce the need for multi-shard view-changes in \oCB{} (at the cost of higher per-transaction latency).

\bibliographystyle{plainurl}
\bibliography{resources}

\end{document}